\documentclass[T1,noinputenc]{cc}
\pdfoutput=1
\usepackage[ascii]{inputenc}
\usepackage[USenglish]{babel}
\usepackage{color,hyperref,url,mathrsfs,bbm,microtype,tikz,xspace,enumitem}
\usepackage[vcentermath]{youngtab}

\DeclareMathOperator{\GL}{GL}
\DeclareMathOperator{\Sym}{Sym}
\DeclareMathOperator{\kron}{Kron}
\DeclareMathOperator{\height}{ht}
\DeclareMathOperator{\poly}{poly}
\newcommand{\C}{\mathbb{C}}
\newcommand{\N}{\mathbb{N}}
\newcommand{\IN}{\mathbb{N}}

\newcommand{\Q}{\mathbb{Q}}

\newcommand{\la}{\lambda}
\newcommand{\Kronecker}{\textsc{Kronecker}\xspace}
\newcommand{\RestrictedKronecker}{\textsc{Restricted Kronecker}\xspace}
\newcommand{\ThreeDMatching}{\textsc{3D Matching}\xspace}
\newcommand{\FourPartition}{\textsc{4-Partition}\xspace}
\newcommand{\ThreePartition}{\textsc{3-Partition}\xspace}
\newcommand{\MachineFlow}{\textsc{Machine Flow}\xspace}
\newcommand{\RNThreeDM}{\textsc{RN3DM}\xspace}
\newcommand{\RNMTS}{\textsc{RNMTS}\xspace}
\newcommand{\Permutation}{\textsc{Permutation}\xspace}
\newcommand{\SpecialConsistency}{\textsc{Special Consistency}\xspace}
\allowdisplaybreaks

\contact{mulmuley@uchicago.edu}
\submitted{8 March 2016}
\title{On vanishing of Kronecker coefficients}
\author{Christian Ikenmeyer\\
Department of Mathematics\\
Texas A\&M University\\
College Station, TX 77843, USA \\
\email{cikenmey@mpi-inf.mpg.de}\\
\homepage{http://people.mpi-inf.mpg.de/~cikenmey/}
\currentaddress
Max Planck Institute for Informatics\\
Saarland Informatics Campus\\
66123 Saarbr\"ucken, Germany\\
\and
Ketan D. Mulmuley\\
Computer Science Department\\
The University of Chicago\\
Chicago, IL 60637, USA\\
\email{mulmuley@uchicago.edu}\\
\homepage{http://gct.cs.uchicago.edu}
\and
Michael Walter\\
Stanford Institute for Theoretical Physics\\
Stanford University\\
Stanford, CA, 94305, USA\\
\email{michael.walter@stanford.edu}\\
\homepage{http://web.stanford.edu/~waltemic/}
}

\begin{abstract}
We show that the problem of deciding positivity  of Kronecker coefficients is NP-hard.
Previously, this problem was conjectured to be  in P, just as for the Littlewood-Richardson coefficients. Our result establishes in a formal way that Kronecker coefficients
are more difficult than Littlewood-Richardson coefficients, unless P $=$  NP.

We also show that there exists a $\# P$-formula for a particular subclass of Kronecker coefficients whose positivity is NP-hard to decide. This is an evidence that, despite the hardness of the positivity problem, there may well exist a positive combinatorial formula for the Kronecker coefficients. Finding such a formula is a major open problem in representation theory and algebraic combinatorics.

Finally, we consider the existence of the partition triples $(\lambda, \mu, \pi)$
such that the Kronecker coefficient $k^\lambda_{\mu, \pi} = 0$
but the Kronecker coefficient $k^{l \lambda}_{l \mu, l \pi} > 0$ for some integer $l>1$.
Such ``holes'' are of great interest as they witness the failure of the saturation property for
the Kronecker coefficients, which is still poorly understood.
Using insight from computational complexity theory, we turn our hardness proof into a positive result: We show that not only do there exist many such triples, but they can also be found efficiently.
Specifically, we show that, for any $0<\epsilon  \le 1$, there exists $0<a<1$ such that, for all $m$,
there exist $\Omega(2^{m^a})$ partition triples $(\lambda,\mu,\mu)$
in the Kronecker cone such that:
(a) the Kronecker coefficient $k^\lambda_{\mu,\mu}$ is zero, (b) the height of $\mu$ is $m$,
(c) the height of $\lambda$ is $\le m^\epsilon$, and (d) $|\lambda|=|\mu| \le  m^3$.
The proof of the last result illustrates the effectiveness of the explicit proof strategy of GCT.
\end{abstract}


\begin{subject}
68Q17, 05E10.
\end{subject}

\begin{document}

\section{Introduction}\label{sintro}
One class of representation-theoretic obstructions in the context of the geometric complexity theory (GCT) approach to the permanent vs.\ determinant problem \citep{GCT2,GCTexplicit} is based on the existence of vanishing rectangular Kronecker coefficients \citep{GCT2,burgnonvanish,kumar}.
These are called occurrence-based obstructions, as opposed to the more general multiplicity-based obstructions.
We refer to \cite{landsbergmath,lothringen} for introduction and background.
It is now known that such occurrence-based obstructions based on vanishing of Kronecker coefficients cannot be used for proving superpolynomial lower bounds for the permanent \citep{IP}.
However, they may still be  useful for proving modest polynomial lower bounds.
The partition triples associated with the rectangular Kronecker coefficients lie in the moment cone \citep{kirwan} associated with the Kronecker coefficients, called the Kronecker cone \citep{burgnonvanish,kumar}.
As pointed out in \citet{kumar}, this makes the problem of showing the existence of such partition triples rather challenging, since the asymptotic techniques of
algebraic geometry and representation theory, such as the ones based on the effective descriptions of the linear inequalities defining the Kronecker cone \citep{berenstein,klyachko,ressayre,verwal}, cannot be used to prove this existence.

The main result in this article (\ref{tkron1}) establishes the existence of a superpolynomial number of partition triples with vanishing Kronecker coefficients, in the Kronecker cone for the given partition size, and satisfying a relaxed form of the additional shape restrictions that arise in GCT.

Its proof, based on the explicit proof strategy of GCT \citep{GCTexplicit,GCTflip,GCT6}, also yields results concerning the complexity of Kronecker coefficients that are of independent interest.
The first such result (\ref{tkronnp}) shows that the problem of deciding positivity of Kronecker coefficients is NP-hard.
The second result (\ref{tkronpos}) gives the first known instance of a positive ($\#P$) formula for a subclass of Kronecker coefficients whose positivity is NP-hard to decide.

We now state these results in more detail after the following preliminary section.

\subsection{Preliminaries in algebraic combinatorics and representation theory}\label{subsec:prelims}

An \emph{partition} $\la$ is defined to be a finite nonincreasing sequence of positive integers $\la = (\la_1, \la_2, \ldots, \la_{\ell})$.
We say that $\la$ has $\ell$ nonzero parts and define the \emph{height} $\height(\lambda)$ of $\la$ to be $\ell$.
We define $\la_i := 0$ for all $i > \ell$ and set $|\la|:=\sum_i \la_i$.
To each partition we associate a so-called \emph{Young diagram}, which is a left-top-justified array of boxes in which the $i$th row contains exactly $\la_i$ boxes.
For example, the Young diagram to the partition $(3,1)$ is {\tiny$\yng(3,1)$}.
We often identify Young diagrams with their partitions and say that $\la$ has $|\la|$ many boxes.
Transposing a Young diagram at the main diagonal gives another Young diagram and we call the corresponding partition the \emph{transpose} partition of~$\la$,
denoted by $\lambda^T$. For example, $(3,1)^T=(2,1,1)$, because transposing $(3,1)$ gives the Young diagram {\tiny$\yng(2,1,1)$}.

When we encode partitions as bit strings there are two fundamentally different ways of doing it:
As a list of numbers in binary or as a list of numbers in unary.
Note that in unary transposing a partition does not significantly change its encoding size, but in binary the 1-row partition $(n)$ can be encoded using $O(\log n)$ bits, while $(n)^T=(1,1,\ldots,1)$ requires $O(n)$ bits.
We will mostly encode partitions in unary.

It is natural to interpret partitions as vectors with integer entries, so that we have a well-defined addition and scalar multiplication with nonnegative integers.
Moreover, dividing a partition by an integer results in a vector with rational entries.

Let $G:=\GL_r$ denote the general linear group, i.e., the group of invertible $r \times r$ matrices.
Let $V$ be a finite dimensional vector space and let $\GL(V)$ denote the set of linear isomorphisms of $V$.
A group homomorphism $\varrho: G \to \GL(V)$ is called a \emph{representation} of $G$.
We say that $V$ is a representation if $\varrho$ is clear from the context.
We say that $G$ \emph{acts linearly} on $V$ and use the short notation $gv := (\varrho(g))(v)$ for $g \in G$, $v \in V$.
If all the coordinate functions of $\varrho$ are given by multivariate polynomials in the $r^2$ coordinate variables of $\GL_r$, then we call $\varrho$ a \emph{polynomial representation}.

A linear subspace $W \subseteq V$ that satisfies $\forall w \in W, \, g \in G: g w \in W$, is called a \emph{subrepresentation}.
Subrepresentations of polynomial representations are always polynomial.
For every representation $V$, the zero vector space and $V$ itself are two subrepresentations.
If $V$ has only these two subrepresentations, then $V$ is called \emph{irreducible}.
Given two representations $(V,\rho_V)$ and $(W,\rho_W)$, then a linear map $\varphi: V \to W$
is called \emph{equivariant} if $g\varphi(v) = \varphi(gv)$ for all $g \in G$, $v \in V$.
If $\varphi$ is an equivariant isomorphism of vector spaces, then
$\varphi$ is called a $G$-isomorphism and the representations $V$ and $W$ are called \emph{isomorphic}.
The different types of isomorphic irreducible polynomial representations of $G$ have been classified completely: They are indexed by partitions of height at most $r$.
In a representation $V$ the sum of all subrepresentations of type $\la$ is called the \emph{$\la$-isotypic component}.
Every representation decomposes into a direct sum of isotypic components.

A vector $v \in V$ that is rescaled by diagonal matrices is called a \emph{weight vector},
i.e., if $\mathrm{diag}(\alpha_1,\ldots,\alpha_r)v = \alpha_1^{\lambda_1}\cdots \alpha_r^{\lambda_r}v$,
then $v$ is called a weight vector of weight $\la$.
For this definition it is not necessary that $\la$ is nonincreasing.
The weight vectors of weight $\la$ form a vector space which is called the \emph{weight subspace of $V$ of weight $\la$}.
A representation decomposes into a direct sum of weight spaces.
In each irreducible representation of type $\la$ there is a unique line of weight $\la$.
This line is called the \emph{highest weight line}.
The highest weight line is also characterized by the Lie algebra action as follows.
For a representation $V$ let $v \in V$ and $A \in \mathfrak g := \C^{r \times r}$.
Note that $\varepsilon A+\mathrm{Id}_r \in G$ for small $\varepsilon$, where $\mathrm{Id}_r$ is the $r \times r$ identity matrix.
The Lie algebra action of $A$ on $v$ is defined as
$A.v := \lim_{\varepsilon\to 0} \varepsilon^{-1}((\varepsilon A+\textup{Id})v - v)$.
Clearly $A.v \in V$.
If we pick $A=E_{i,j}$ to be the matrix that has a single 1 at position $(i,j)$ with $i<j$ and zeros everywhere else, then the map $v\mapsto E_{i,j}.v$ is called a \emph{raising operator}.
A weight vector that is mapped to zero by all raising operators is called a \emph{highest weight vector}.
Each irreducible representation has a unique line of highest weight vectors and their weight determines the type of the irreducible representation.

The \emph{multiplicity} of the type $\la$ in a representation $V$ is the dimension of the vector space of highest weight vectors of type~$\la$.
If we decompose $V$ into a direct sum of irreducibles, then this multiplicity counts how often a copy of type $\la$ appears in the decomposition.

If we have $k$ commuting actions of several copies of $G$ on $V$, then we use the representation theory of the cartesian powers $G^k$, which is very similar to the representation theory of $G$.
We will mainly be concerned with $k=2$ and $k=3$.
The types of irreducible representations of $G^k$ are given by $k$-tuples of partitions, weight vectors are defined by their scaling behavior under $k$-tuples of diagonal matrices,
and the Lie algebra action is defined via $k$-tuples of matrices, where the raising operators are $k$-tuples of matrices in which only one matrix is nonzero.
Irreducible representations of $G$ are called \emph{Weyl modules}, while irreducible representations of $G^k$ are isomorphic to a $k$-fold tensor product of Weyl modules.

Using the group homomorphism $\GL_r \times \GL_r \to \GL_{r^2}$, $(g,g')\mapsto g \otimes g'$, every representation of $\GL_{r^2}$ is also a representation of $\GL_r \times \GL_r$.
For a partition $\la$ let $V_\la(\GL_{r^2})$ be the Weyl module of type $\la$.
Even though $V_\la(\GL_{r^2})$ is an irreducible $\GL_{r^2}$ representation, $V_\la(\GL_{r^2})$ decomposes into a nontrivial direct sum of isotypic components with respect to the action of $\GL_r \times \GL_r$.
For partitions $\mu$ and $\pi$ the multiplicity of the irreducible $V_\mu(\GL_r) \otimes V_\pi(\GL_r)$ in $V_\la(\GL_{r^2})$ is called the \emph{Kronecker coefficient} $k^\la_{\mu,\pi} \in \mathbb N$ \citep{fultonrepr}.
If $\height(\la)\leq r^2$, $\height(\mu)\leq r$, and $\height(\pi)\leq r$, then $k^\la_{\mu,\pi}$ does not depend on $r$.
Otherwise $k^\la_{\mu,\pi}$=0.

The Kronecker coefficient arises as a multiplicity in several other representation theoretic decompositions, for example as
the multiplicity of $V_\la(\GL_r) \otimes V_\mu(\GL_r) \otimes V_\pi(\GL_r)$ in $V_{(n)}(\GL_{r^3})$ via the group homomorphism $\GL_r^3 \to \GL_{r^3}$, $(g,g',g'')\mapsto g \otimes g' \otimes g''$,
where $\la$, $\mu$, and $\pi$ have exactly $n$ boxes.

For $k^\la_{\mu,\pi}$ to be positive it is required that $\la$, $\mu$, and $\pi$ are partitions of the same number, i.e., $|\la|=|\mu|=|\pi|=n$ for some $n$.
This implies that if $k^\la_{\mu,\pi}>0$, then the rescaled partitions $\la/n$, $\mu/n$, and $\pi/n$ are three discrete probability distributions.
Another necessary condition for $k^\la_{\mu,\pi}>0$ is $\height(\la)\leq\height(\mu)\cdot\height(\pi)$.
The coefficient is invariant under permuting the three parameters, so
\begin{equation}\label{eq:permutation symmetry}
  k^\la_{\mu,\pi}=k^\la_{\pi,\mu}=k^\mu_{\pi,\la}
\end{equation}
Moreover, transposing any two of the three parameters does not change the coefficient:
\begin{equation}\label{eq:transpose symmetry}
  k^\la_{\mu,\pi}=k^\la_{\mu^T,\pi^T}
\end{equation}

Another important description of the Kronecker coefficient is presented in \ref{lem:antisymmetric}:
The tensor power $\otimes^3 \C^r = (\mathbb C^r)^{\otimes 3}$ has a canonical action of $G^3$ via
$(g,g',g'')(u\otimes v\otimes w) = (g u)\otimes(g'v)\otimes(g''w)$ for all
$(g,g',g'')\in G^3$ and $u\otimes v\otimes w$ a rank 1 tensor from $(\mathbb C^r)^{\otimes 3}$, where the action is defined by linear continuation.
This action induces an action on the $n$-th tensor power $\otimes^n((\mathbb C^r)^{\otimes 3})$ via
$(g,g',g'') ((u_1\otimes v_1 \otimes w_1) \otimes (u_2\otimes v_2 \otimes w_2) \otimes \cdots \otimes (u_n\otimes v_n \otimes w_n)) :=
((g u_1\otimes g' v_1 \otimes g'' w_1) \otimes (g u_2\otimes g' v_2 \otimes g'' w_2) \otimes \cdots \otimes (g u_n\otimes g' v_n \otimes g'' w_n))$.
The antisymmetric tensors form a subrepresentation: $\bigwedge^n (\mathbb C^r)^{\otimes 3} \subseteq \otimes^n (\mathbb C^r)^{\otimes 3}$.
The Kronecker coefficient $k^\lambda_{\mu,\pi}$ is equal to the multiplicity of the irreducible representation $V_{\lambda^T}(\GL_r) \otimes V_{\mu^T}(\GL_r) \otimes V_{\pi^T}(\GL_r)$ of $\GL_r^3$ in $\bigwedge^n (\mathbb C^r)^{\otimes 3}$,
provided $r$ is large enough.

Given two partition triples $(\la,\mu,\pi)$ and $(\la',\mu',\pi')$ such that $k^\la_{\mu,\pi}>0$ and $k^{\la'}_{\mu',\pi'}>0$,
then $k^{\la+\la'}_{\mu+\mu',\pi+\pi'}>0$.
This is called the \emph{semigroup property}.
The convex cone defined by
\[ \kron(r) := \operatorname{cone} \, \{(\lambda,\mu,\pi) \mid \height(\la),\height(\mu),\height(\pi)\leq r, \, k^\la_{\mu,\pi}>0 \} \]
is a polyhedral cone \citep{kirwan}, called the \emph{Kronecker cone}.
Here we think of $(\lambda,\mu,\pi)$ as a vector in $\Q^{3 r}$.
A partition triple outside of this cone trivially has a zero Kronecker coefficient.

Finding a combinatorial description of $k^\la_{\mu,\pi}$ is an important outstanding problem (see \ref{subsec:sharpP} below).
Only for some special cases is a combinatorial description known, for example for the \emph{Littlewood-Richardson coefficients}.
The Littlewood-Richardson coefficients are those $k^\la_{\mu,\pi}$ for which $\la$, $\mu$, and $\pi$ have a sufficiently long first row
such that $|\overline\la|=|\overline\mu|+|\overline\pi|$, where $\overline\la$ is the partition $\la$ with its longest row removed.
The positivity of Littlewood-Richardson coefficients can be decided in strongly polynomial time \citep{knutsontao2,GCT3}.

Another important subclass of Kronecker coefficients are the \emph{rectangular Kronecker coefficients}.
For a given partition $\lambda$ with size $|\lambda|$ divisible by $r$, let $\delta(\lambda)$ denote the rectangular partition $(d,\ldots,d)$ ($r$ times), where $d=|\lambda|/r$.
We call the Kronecker coefficient $k^\lambda_{\delta(\lambda), \delta(\lambda)}$ \emph{rectangular}.
Rectangular Kronecker coefficients play a special role in geometric complexity theory (see \ref{sexceptional} below).

\subsection{NP-hardness of deciding positivity of Kronecker coefficients}
Let \Kronecker be the problem of deciding positivity of $k_{\mu,\pi}^\lambda$, given as input the three partitions $\lambda$, $\mu$, and $\pi$ in unary.
This problem is of fundamental interest in the context of the explicit proof strategy of GCT \citep{GCTexplicit,GCTflip,GCT6}.
Our first result is the following:

\begin{theorem}\label{tkronnp}
\Kronecker is NP-hard.
\end{theorem}

It was conjectured in~\citet{GCT6} that the problem of deciding positivity of Kronecker coefficients is in $P$.
\ref{tkronnp} shows that this is not so, in general, assuming that $P\not = NP$.
This is in contrast to the special case of the Littlewood-Richardson coefficients, where positivity can be decided in strongly polynomial time, as explained above.

\subsection{\texorpdfstring{A $\#P$-formula for a subclass of partitions of type NP}{A \#P-formula for a subclass of partitions of type NP}}\label{subsec:sharpP}
To find a positive formula for Kronecker coefficients ``akin to'' the well known positive Littlewood-Richardson rule is an unsolved problem in classical representation theory.
We refer to \cite{stanley} for the history and importance of this problem, where it is listed as one of the twenty-five ``outstanding open problems''.
In classical representation theory, the phrase ``akin to'' is used only informally.
A formal complexity-theoretic version of this problem is to find a $\#P$-formula for Kronecker coefficients.  By a $\#P$-formula for the Kronecker coefficient
$k^\lambda_{\mu,\pi}$,  we mean a formula of the form:

\[ k^\lambda_{\mu,\pi}= \sum_{\sigma \in
\{0,1\}^{p(\langle \lambda \rangle, \langle \mu \rangle, \langle \pi \rangle)}}
F(\lambda,\mu,\pi,\sigma), \]

where, for a partition $\lambda=(\lambda_1,\lambda_2, \ldots, \lambda_l)$,
$\langle \lambda \rangle$  denotes
the total bitlength of the  specification of  $\lambda_j$'s in binary,
$p(\langle \lambda \rangle, \langle \mu \rangle, \langle \pi \rangle)$
is a polynomially-bounded function of the bit-lengths
$\langle \lambda \rangle, \langle \mu \rangle$, and $\langle \pi \rangle$,
and $F(\lambda,\mu,\pi,\sigma)$ is a polynomial-time-computable $0$-$1$ function
of $\lambda,\mu, \pi$, and the bit-string $\sigma$.
By a positive formula, we mean a $\#P$-formula henceforth.

\begin{definition}
Let $\Pi$ be a class of partition triples.
We say that $\Pi$ is of \emph{type NP} if the problem of deciding positivity of $k^\lambda_{\mu,\pi}$, with
$(\lambda,\mu,\pi) \in \Pi$, is NP-hard. (The problem mentioned here is
a promise problem. That is, we are promised that the input triple is in the
subclass $\Pi$.)
Likewise, we say that $\Pi$ is of \emph{type P} if the problem of deciding positivity of $k^\lambda_{\mu,\pi}$, with $(\lambda,\mu,\pi) \in \Pi$, is in~P.
\end{definition}

All positive rules known so far for restricted classes of Kronecker coefficients have been for subclasses of partition triples that are either known or conjectured to be of type P.
For example, the classical Littlewood-Richardson rule gives a positive rule for Littlewood-Richardson coefficients, which, as
already mentioned, constitute a special class of Kronecker coefficients.
The corresponding subclass of partition triples is of type P, since the problem of deciding positivity of Littlewood-Richardson coefficients is in $P$ \citep{knutsontao2,GCT3}.
\cite{GCT4} give a positive rule for Kronecker coefficients when two of the partitions have height at most two.
The corresponding subclass of partition triples is of type P, since the Kronecker coefficient can be computed in this case (and more generally, for
partitions of bounded height) in polynomial time (\cite{christwalter}; \cite{newalgorithm}).
\cite{blasiak} gives a positive rule for Kronecker coefficients when one of the partitions is a hook.
The corresponding subclass of partition triples is conjectured to be of type P, since the problem of deciding positivity of Kronecker coefficients, when one of the partitions is a hook, is believed to be in $P$ (in view of \ref{thook}).

The following result gives the first known instance of a positive
rule for Kronecker coefficients for a subclass of partition triples of type NP.

\begin{theorem}\label{tkronpos}
There exists a $\#P$-formula for Kronecker coefficients for
a  subclass of partition triples of type NP. Here the partition triples can be specified in unary or binary.
\end{theorem}

The proof of this result exhibits an explicit such subclass of partition triples of type NP (see \ref{sbounds} and \ref{snppos}).

\ref{tkronpos} provides good evidence in support of the conjecture in \cite{GCT6} that there exists a $\#P$-formula for Kronecker coefficients in general.
This would in particular imply that \Kronecker is in NP, which is not known so far.

\subsection{Exceptional Kronecker coefficients}\label{sexceptional}

In order for a Kronecker coefficient $k_{\mu,\pi}^\lambda$ to be useful for proving a polynomial lower bound for the permanent, the partition triple must have a number of exceptional properties \citep{GCT2,landsbergmath}.
This is captured by the following definition:

\begin{definition}\label{dexceptional}
Fix any constant $0 < \epsilon \le 1$, and
a constant $b>1$.
We call a partition triple $(\lambda,\mu,\pi)$ with  $\lvert\la\rvert=\lvert\mu\rvert=\lvert\pi\rvert$ \emph{$(\epsilon,b)$-exceptional} if:
\begin{enumerate}
\item $k_{\mu,\pi}^\lambda = 0$,
\item\label{dexceptional:1} $\mu=\pi=\delta(\lambda)$, with  $|\lambda|=|\mu|=|\pi|$  divisible by
$r:=\height(\mu)=\height(\pi)$,
\item\label{dexceptional:2} $\height(\lambda) \le r^\epsilon$,
\item\label{dexceptional:3} $(\lambda,\mu,\pi) \in \kron(r)$,
\item\label{dexceptional:4} $|\lambda|=|\mu|=|\pi| \le r^b$,
\item\label{dexceptional:5} the multiplicity $p(\lambda)$ of the Weyl module $V_\lambda(GL_{r^2}(\C))$  in $\Sym^d (\Sym^r(\C^{r^2}))$,
$d=|\lambda|/r$,
is positive, and
\item\label{dexceptional:6} $\lambda_0 \geq |\lambda| (1 - r^{\epsilon/2-1})$.
\end{enumerate}
We also call a partition tuple merely \emph{exceptional}, without mentioning $\epsilon$ and $b$, if it is understood that $\epsilon$ can be chosen to be arbitrarily small, with $b$ a large enough constant depending on $\epsilon$, and $r \rightarrow \infty$.
\end{definition}

By~\cite{burgnonvanish}, \ref{dexceptional:1} implies~\ref{dexceptional:3}, assuming that the height of
$\lambda$ is $\le r^2$, which is so by~\ref{dexceptional:2}.

The constraint~\ref{dexceptional:3} is significant.
Proving existence of the partition triples as in \ref{dexceptional} is delicate because of this constraint.
Indeed, it may be possible to prove existence of superpolynomially many partition triples satisfying the
constraints other than~\ref{dexceptional:1} and~\ref{dexceptional:3} using the known linear inequalities defining the Kronecker cone \citep{berenstein,klyachko,ressayre,verwal}.
But the constraint~\ref{dexceptional:3} implies that such asymptotic techniques based on the description of the Kronecker cone cannot be used to demonstrate existence of partition triples as in \ref{dexceptional}.
This is the main significance of the results in \cite{burgnonvanish,kumar}.

By the Saturation Theorem \citep{knutsontao1,derksen}, the Littlewood-Richardson coefficients cannot vanish for the partition triples that lie in the analogously defined Littlewood-Richardson cone.
The constraint~\ref{dexceptional:3} also implies that in order to prove existence of the partition triples as in \ref{dexceptional}, one needs to understand the failure of the saturation property for the Kronecker coefficients in one way or another.

The constraint~\ref{dexceptional:6} is motivated by~\cite{landsberg}.
There, it is shown that this condition holds if $V_\lambda(G)$ is a re\-pre\-sen\-ta\-tion-the\-o\-re\-tic obstruction \citep{GCT2}.

It is a priori not at all clear that for any given constant $0<\epsilon \le 1$ and a large enough constant $b>1$ depending on $\epsilon$,
exceptional partition triples exist for arbitrary $r$.
The experimental evidence in \cite{iken} for small values of $r$ (with suitable $\epsilon$ and $b$) suggests that they are very rare, though they do exist for these small values.
In summary, although their density can be expected to be extremely small, it is a relevant and rather non-trivial problem in the context of GCT to show that exceptional partition triples exist and that their number is large enough.

\subsection{Construction of superpolynomially many partition tri\-ples in the Kronecker cone with vanishing Kronecker coefficients}\label{suncond}
As the first step towards this goal, we relax the condition~\ref{dexceptional:1} to the weaker requirement that only $\mu=\pi$, the condition \ref{dexceptional:5} to the weaker requirement weaker that $\lambda$ is not a hook (since it can be shown that $p(\lambda)=0$ if $\lambda$ is a hook), and ignore the condition~\ref{dexceptional:6}.
A priori, it is not clear that partition triples with these properties exist even after this shape relaxation, since condition~\ref{dexceptional:3} is retained.
The following result shows that the number of Kronecker coefficients with this relaxation of \ref{dexceptional} is superpolynomial.

\begin{theorem}[The main result]\label{tkron1}
For any $0 < \epsilon \le 1$, there exists
 $0<a<1$, such that, for all   $m$,
there exist  $\Omega(2^{m^a})$  partition triples $(\lambda,\mu,\pi)$ such that
\begin{enumerate}
\item\label{tkron1:0} $k_{\mu,\pi}^\lambda = 0$,
\item\label{tkron1:1} $\mu=\pi$,
\item\label{tkron1:2} $\height(\mu) = m$, and   $\height(\lambda) \le m^ \epsilon$,
\item\label{tkron1:3} $(\lambda,\mu,\pi) \in \kron(m)$,
\item $|\lambda|=|\mu|=|\pi| \le  m^3$,  and
\item\label{tkron1:5} $\lambda$ is not a hook.
\end{enumerate}
\end{theorem}

Furthermore:

\begin{theorem}\label{tkron11}
Assuming coNP $\neq$ NP, the set of partition triples satisfying constraints \ref{tkron1:0}--\short\ref{tkron1:5} as well as
\begin{enumerate}[start=7]
\item\label{tkron11:6} $(\lambda^T, \mu^T, \mu) \in \kron(m')$, where $m'$ is the maximum of the heights of $\lambda^T$, $\mu^T$ and $\mu$, and
\item\label{tkron11:7} $(\lambda,\mu^T,\mu^T) \in \kron(m'')$,
where $m''$ is the maximum of the heights of $\lambda$ and  $\mu^T$,
\end{enumerate}
is superpolynomial in $m$, as $m \rightarrow \infty$.
\end{theorem}

The constraints~\short\ref{tkron1:3}, \ref{tkron11:6}, and \ref{tkron11:7} together guarantee that the vanishing of $k^\lambda_{\mu,\mu}$ cannot be directly shown using the defining inequalities of the Kronecker cone \citep{berenstein,klyachko,ressayre,verwal} in conjunction with the known symmetries~\ref{eq:permutation symmetry} and~\ref{eq:transpose symmetry}.

Our proof of \ref{tkron1} shows that the partition triples $(\lambda,\mu,\mu)$ satisfying the constraints therein can even be constructed \emph{explicitly}.
This means there is a one-to-one map from the set of Boolean strings of length $\le m^a$ to the set of partitions triples $(\lambda,\mu,\mu)$ with properties~\short\ref{tkron1:0}--\short\ref{tkron1:5} that can be computed in $\poly(m)$ time.

While \ref{tkron1} shows existence of superpolynomially many partition triples satisfying the constraints therein, the density of such partition triples is exponentially small, since $a$ therein is much smaller than $1$ (see \ref{sexample}).
This may explain why vanishing Kronecker coefficients with the partition triples in the Kronecker cone occur so rarely in computer experiments, as observed in \cite{iken}.

\subsection{Proof technique}\label{sprooftech}
\ref{tkronnp} is proved by extending the NP-completeness technique in \cite{brunetti} in conjunction with
the fundamental lower and upper bounds on Kronecker coefficients established in \cite{manivel1997applications,burgisserexp,vallejo2000plane}.
\ref{tkronpos} is a byproduct of this proof.

A refined form of \ref{tkronnp} lies at the heart of the proof of \ref{tkron1}.
Specifically, we show in \ref{tnprefined} that the problem of deciding positivity of $k^\lambda_{\mu,\pi}$ remains NP-hard under polynomial-time many-one reductions \citep{karpcomplexity} even when the partitions $(\lambda,\mu,\pi)$ are required to satisfy the constraints \ref{tkron1:1}--\short\ref{tkron1:5}.
This is done by extending the proof technique of \ref{tkron1} using the result in \cite{burgnonvanish} that $(\lambda,\delta(\lambda), \delta(\lambda))$, for $|\lambda|$ divisible by $r$, lies in the Kronecker cone whenever the height of $\lambda$ is $\le r^2$.
By~\cite{fortune}, if there exists a co-sparse NP-complete language under polynomial-time many-one-reductions, then P=NP.
(Here we call a language {\em sparse} if the number of strings in it
of bitlength $\le N$ is bounded by a fixed polynomial in $N$.
It is called {\em co-sparse} if its complement is sparse.)
Hence \ref{tnprefined} in conjunction with \cite{fortune} implies that the set of partition triples satisfying the constraints \ref{tkron1:0}--\short\ref{tkron1:5} is non-sparse, i.e., has size superpolynomial in $m$, assuming that P$\neq$NP.

To prove \ref{tkron1}, we have to discard of the assumption that P=NP and replace the superpolynomial bound by $\Omega(2^{m^a})$ bound for some $a>0$.
This is done in \ref{tkron1crux} by exhibiting
 a  polynomial-time one-one reduction from the \ThreeDMatching problem \citep{johnson} to the problem of deciding positivity of Kronecker coefficients, with the partition triples satisfying the constraints~\ref{tkron1:1}--\short\ref{tkron1:5},
where  a {\em polynomial-time one-one reduction}
means an injective polynomial-time many-one reduction.
Hence, the $\Omega(2^{m^a})$ bound in \ref{tkron1} follows from a similar lower bound on the number of instances of the \ThreeDMatching problem with a ``NO'' answer.
The proof automatically shows that $\Omega(2^{m^a})$ partition triples satisfying the constraints in \ref{tkron1} can be constructed explicitly.
This follows by fixing a suitable set of $2^{N^b}$ instances, for some constant $b>0$, of the \ThreeDMatching problem of bitlength $\le N$ with ``NO'' answer, and mapping them injectively, via a sequence of polynomial time one-one reductions, to $\Omega(2^{m^a})$ such partition triples.

\ref{tkron11} is proved by extending the proof of \ref{tkron1} using an auxiliary result, \ref{lsparse}, which extends the hardness vs.~non-sparseness result in \cite{fortune}, together with the result in~\cite{BCMW} which asserts that the membership problem for the Kronecker cone is in $\mbox{NP} \cap \mbox{coNP}$.

\subsection{Effectiveness of the explicit proof strategy}
Perhaps the most novel aspect of this paper is the synthesis of the representation theory of Kronecker coefficients with the theory of NP-completeness to prove unconditionally existence of superpolynomially many partition triples in the Kronecker cone with vanishing Kronecker coefficients.

In principle, the existence of partition triples satisfying the constraints in \ref{tkron1} may be proved by a nonconstructive
technique.
Yet, the only way we can prove this existence at present is by constructing such partitions {\em explicitly}, using the theory of algorithms,
as done in the proof of \ref{tkron1}.
Thus this proof illustrates effectiveness of the explicit proof strategy of GCT \citep{GCTexplicit,GCTflip,GCT6} in a nontrivial setting.

\subsection{Organization}
The rest of this article is organized as follows.
\ref{sbounds} describes the lower and upper bounds for the Kronecker coefficients that are needed for the proofs of \ref{tkronnp} and \ref{tkronpos}.
These proofs are given in \ref{snppos}.
A refinement of \ref{tkronnp}, which is needed for the proof of \ref{tkron1}, is proved in \ref{srefined}.
\ref{tkron1} and \ref{tkron11}  are proved in \ref{scone}.
\ref{srectangular} proves additional results in support of the conjecture in \cite{GCT6} that the problem of deciding positivity of rectangular Kronecker coefficients is in $P$.

\section{Lower and upper bounds for the Kronecker coefficient}\label{sbounds}
In this section, we give representation-theoretic proofs of some known lower and upper bounds from \cite{manivel1997applications,burgisserexp,vallejo2000plane} for  the Kronecker coefficients.
These bounds as well as their representation-theoretic interpretation given here will play a crucial role in the proofs of \ref{tkronnp} and \ref{tkronpos} in \ref{snppos}.
We begin with the following well-known result (whose proof we include for the sake of completeness):

\begin{lemma}
\label{lem:antisymmetric}
  Let $\lambda,\mu,\pi$ denote Young diagrams, each with $n$ boxes and no more than $r$ columns.
  Then the Kronecker coefficient~$k^\lambda_{\mu,\pi}$ is equal to the multiplicity of the irreducible representation $V_{\lambda^T}(\GL(r)) \otimes V_{\mu^T}(\GL(r)) \otimes V_{\pi^T}(\GL(r))$ of $\GL(r)^3$ in the anti-symmetric subspace $\bigwedge^n (\mathbb C^r)^{\otimes 3}$.
\end{lemma}
\begin{proof}
Let $\tilde k^\kappa_{\alpha,\beta,\gamma}$ denote the multiplicity of the irreducible representation $V_{\alpha}(\GL(r)) \otimes V_{\beta}(\GL(r)) \otimes V_{\gamma}(\GL(r))$ of $\GL(r)^3$ in the Weyl module $V_\kappa(\GL(r^3))$.
Our original definition of the Kronecker coefficient is easily seen to be equivalent to $\tilde k^{(n)}_{\lambda,\mu,\pi}$ (e.g.,~\cite{walter2014thesis}, but this is standard), and so we need to show that
\begin{equation}\label{eq:antisymmetric claim}
  \tilde k^{(n)}_{\lambda,\mu,\pi} = \tilde k^{(1^n)}_{\lambda^T,\mu^T,\pi^T}.
\end{equation}
Given a partition $\kappa$, let $[\kappa]$ denote the Specht module (i.e., an irreducible representation) of $S_n$.
By Schur-Weyl duality, $\tilde k^\kappa_{\alpha,\beta,\gamma}$ is also equal to the multiplicity of
$[\kappa]$ in the  triple tensor product $[\alpha] \otimes [\beta] \otimes [\gamma]$. Since the
representations of the symmetric group are self-dual, this shows that:
\[ \tilde k^\kappa_{\alpha,\beta,\gamma} = \dim ([\alpha] \otimes [\beta] \otimes [\gamma] \otimes [\kappa])^{S_n} \]
Since  $[\lambda^T] = [\lambda] \otimes [(1^n)]$, $[(1^n)] \otimes [(1^n)] = [(n)]$, and $[(n)]$ is the trivial representation,
\ref{eq:antisymmetric claim} follows at once.
\end{proof}

Given a (finite) point set $P \subseteq \{0,\dots,r-1\}^3$,
let $x_P(i)$, $0 \le i \le r-1$, be the number of points in $P$ with
the $x$-coordinate $i$.
We call $x_P=(x_P(0), \ldots, x_P(r-1))$ the \emph{$x$-marginal} of $P$.
We similarly define the $y$-marginal $y_P$ and the $z$-marginal $z_P$.
The triple $(x_P, y_P, z_P)$ is called the {\em marginals} of $P$.

\begin{definition}
We define $t^\lambda_{\mu,\pi}$ as the number of point sets $P \subseteq \{0,\dots,r-1\}^3$ with marginals
 $(\lambda^T, \mu^T, \pi^T)$.
\end{definition}

Note that $\lambda^T_i$ is the number of boxes in the $i$-th \emph{column} of~$\lambda$.

The coefficients $t^\lambda_{\mu,\pi}$ have a pleasant representation-theoretical interpretation that is closely related to \ref{lem:antisymmetric}.
To see this, observe that we can associate with any point set
\[ P = \{(x_1,y_1,z_1), \dots, (x_n,y_n,z_n)\} \subseteq \{0,\dots,r-1\}^3 \]
of cardinality $n$ the following vector:
\[ \psi_P = \bigwedge_{j=1}^n e_{x_j} \otimes e_{y_j} \otimes e_{z_j} \in \bigwedge^n \bigl (\mathbb C^r)^{\otimes 3}, \]
where $\{e_i\}$ is the standard basis of $\C^r$.
The vectors $\psi_P$ form a basis of $\bigwedge^n (\mathbb C^r)^{\otimes 3}$ as $P$ ranges over all such point sets.
Moreover, each $\psi_P$ is  a weight vector for  the $GL(r)^3$-action, whose weight is given by the marginals of the point set $P$ (see \ref{subsec:prelims} for the definition of a weight vector).
Thus we obtain the following result.

\begin{lemma}
\label{lem:antisymmetric weights}
  Let $\lambda,\mu$, and $\pi$ denote Young diagrams with $n$ boxes each and no more than $r$ columns.
  Then  $t^\lambda_{\mu,\pi}$ is equal to the weight multiplicity of $(\lambda^T,\mu^T,\pi^T)$ in the anti-symmetric
$GL(r)^3$-module $\bigwedge^n (\mathbb C^r)^{\otimes 3}$.
\end{lemma}

Following~\cite{vallejo2000plane}, we call a subset $P \subseteq \{0,\dots,r-1\}^3$ a \emph{pyramid} if,
 for any $(x,y,z) \in P$ and  $0 \leq x' \leq x$, $0 \leq y' \leq y$, $0 \leq z' \leq z$,
we have that $(x',y',z') \in P$. (It would also be natural to call such $P$ a \emph{3-partition}; cf.~\cite{manivel1997applications}.)

\begin{definition}
Let  $p^\lambda_{\mu,\pi}$ denote the number of pyramids with marginals $(\lambda^T,\mu^T,\pi^T)$.
\end{definition}

From our representation-theoretic interpretation, we directly obtain the following fundamental bounds, which were proved previously using different methods in~\cite{manivel1997applications,burgisserexp} \citep[cf.][]{vallejo2000plane}:

\begin{lemma}\label{lbounds}
  For all partitions $\lambda,\mu,\pi$, we have $p^\lambda_{\mu,\pi} \leq k^\lambda_{\mu,\pi} \leq t^\lambda_{\mu,\pi}$.
\end{lemma}
\begin{proof}
The upper bound on $k^\lambda_{\mu,\pi}$ follows directly from \ref{lem:antisymmetric} and \ref{lem:antisymmetric weights}, since the irreducible representations of $GL(r)^3$ in $\bigwedge^n (\mathbb C^r)^{\otimes 3}$ are in one-to-one correspondence with their highest weight vectors, and every highest weight vector is a weight vector.

For the lower bound, suppose that
\[ P= \{(x_1,y_1,z_1), \dots, (x_n,y_n,z_n)\} \subseteq \{0,\dots,r-1\}^3 \]
is a pyramid with marginals $(\lambda^T,\mu^T,\pi^T)$.
We will show that $\psi_P$ is not only a weight vector, but in fact a highest weight vector.
For this, we need to argue that $\psi$ is annihilated by all raising operators (cf.~\ref{subsec:prelims}).
Thus consider $(E_{x',x},0,0)$, where $E_{x',x}$ denotes the upper triangular matrix with a single 1 in the $x'$-th row and $x$-th column, and otherwise zero (here $x' < x$).
Its action on $\psi_P$ is given by
\begin{align*}
  &\quad (E_{x',x},0,0) \cdot \psi_P \\
  &= \sum_{j=1}^n (-1)^{j-1} E_{x',x} e_{x_j} \otimes e_{y_j} \otimes e_{z_j} \wedge \bigwedge_{j'=1, j' \neq j}^n e_{x_{j'}} \otimes e_{y_{j'}} \otimes e_{z_{j'}} \\
  &= \sum_{j=1}^n (-1)^{j-1} \delta_{x,x_j} e_{x'} \otimes e_{y_j} \otimes e_{z_j} \wedge \bigwedge_{j'=1, j' \neq j}^n e_{x_{j'}} \otimes e_{y_{j'}} \otimes e_{z_{j'}} = 0,
\end{align*}
since each summand vanishes individually.
Indeed, if $x \neq x_j$ then $\delta_{x,x_j} = 0$ and so the summand is zero.
Otherwise, if $x = x_j$ then $(x_j,y_j,z_j) \in P$ and $x' < x$ imply that $(x',y_j,z_j) \in P$ by the pyramid condition; therefore $e_{x'} \otimes e_{y_j} \otimes e_{z_j}$ appears twice in the wedge product and so the summand vanishes as well.
The same argument applies to the other generators $(0,E_{y',y},0)$ and $(0,0,E_{z',z})$ of $\mathfrak n$.
Thus we conclude that the pyramid condition  ensures that $\psi_P$ is a highest weight vector.
\end{proof}

\begin{corollary}
\label{cor:upper and lower match}
  Let $\lambda$, $\mu$, $\pi$ be partitions such that any point set with marginals $(\lambda^T,\mu^T,\pi^T)$ is necessarily a pyramid. Then $k^\lambda_{\mu,\pi} = t^\lambda_{\mu,\pi} = p^\lambda_{\mu,\pi}$.
\end{corollary}

\section{Kronecker coefficients with \#P-formulae}\label{snppos}
In this section, we prove \ref{tkronnp} and \ref{tkronpos}.

For this, we first derive a sufficient condition on the marginals $(\lambda^T, \mu^T, \pi^T)$ such that any compatible point set is necessarily a pyramid (and hence \ref{cor:upper and lower match} is applicable).
Adapting the approach of \cite{brunetti}, we consider a point set $P$ such that $P_r \subseteq P \subsetneq P_{r+1}$, where
\[ P_r = \{ (x,y,z) \in \{0,\dots,r-1\}^3 : x+y+z \leq r-1 \} \]
denotes the \emph{simplex} of side length $r \geq 1$.
Let $n$ denote the total number of points in $P$.
Then the projection of the \emph{barycenter} $b_P := \sum_{p \in P} p$ of $P$ onto the diagonal $(1,1,1)$ can be computed as follows:
\begin{equation}\label{eq:barycenter definition}
    b_P \cdot (1,1,1)
  = b_r \cdot (1,1,1) + r  (n - \lvert P_r \rvert)
  =: p(n),
\end{equation}
where $b_r$ denotes the barycenter of the simplex $P_r$.
Note that this formula depends only $n$, the number of points in the point set $P$. 
We can thus define a function $p(n)$ by~\ref{eq:barycenter definition}, first for all $n$ such that $\lvert P_r \rvert \leq n < \lvert P_{r+1} \rvert$, and then, by varying $r$, for all $n$.
Explicitly,
\begin{equation}\label{eq:barycenter explicit}
  p(n) = b_{r(n)} \cdot (1,1,1) + r(n)  \bigl( n - \lvert P_{r(n)} \rvert\bigr),
\end{equation}
where $r(n)$ is the maximal $r$ such that $\lvert P_r \rvert = r(r+1)(r+2)/6 \leq n$.

\begin{definition}
Let us call $(\lambda,\mu,\pi)$, with $|\lambda|=|\mu|=|\pi|=n \not = 0$,
 \emph{simplex-like} if there exists some $r$ such that the Young diagrams of $\lambda, \mu$, and $\pi$
have at most $r+1$ columns, and
\[ \sum_{i=0}^r i \lambda^T_i + \sum_{j=0}^r j \mu^T_j + \sum_{k=0}^r k \pi^T_k = p(n). \]
\end{definition}
Whether $(\lambda,\mu,\pi)$ is simplex-like can be checked in polynomial time (even assuming that $\lambda,\mu$ and $\pi$ are given in binary).

The following lemma justifies the term ``simplex-like''.

\begin{lemma}\label{lsimpyr}
  Let $(\lambda,\mu,\pi)$ be simplex-like.
  Then any point set $P$ with marginals $(\lambda^T,\mu^T,\pi^T)$ is necessarily of the form $P_r \subseteq P \subsetneq P_{r+1}$,
for some $r \geq 1$.
  In particular, $P$ is a pyramid.
\end{lemma}
\begin{proof}
  The level sets of the function $p \mapsto p \cdot (1,1,1)$ restricted to the positive octant are precisely the faces $\{(x,y,z) \in \mathbb N^3 : x+y+z = k-1\}$ of the simplices $P_k$.
  Thus it is geometrically obvious that for an arbitrary point set $P$ with $n$ elements, $b_P \cdot (1,1,1)$ is never smaller than $p(n)$, and that it attains this minimum if and only if $P_r \subseteq P \subsetneq P_{r+1}$, for some $r \geq 1$ \citep[cf.][]{brunetti}.
  Therefore, it suffices to show that our assumptions imply that $b_P \cdot (1,1,1) = p(n)$.
  This is indeed true as shown in the following computation, which relies on the fact that the barycenter is purely a function of the marginals:
  \[
    b_P \cdot (1,1,1)
  = \!\!\!\!\!\sum_{(x,y,z) \in P} \!\!\! (x+y+z)
  = \sum_{i=0}^r i \lambda^T_i + \sum_{j=0}^r j \mu^T_j + \sum_{k=0}^r k \pi^T_k
  = p(n).
  \]
The last step follows because  $(\lambda,\mu,\pi)$ is simplex-like.
\end{proof}

\begin{theorem}
\label{the:simplex-like}
  Let $(\lambda,\mu,\pi)$ be simplex-like. Then $k^\lambda_{\mu,\pi} = t^\lambda_{\mu,\pi} = p^\lambda_{\mu,\pi}$.
  In particular, this family of Kronecker coefficients has a $\#P$-formula. Here, $\lambda,\mu$ and $\pi$ can be given in unary or binary.
\end{theorem}
\begin{proof} This follows from \ref{lsimpyr}, \ref{cor:upper and lower match},
and the fact that $t^\lambda_{\mu,\pi}$ has a $\#P$-formula.
The last assertion follows because  the bit-length of the unary specification of a simplex-like
partition triple is polynomial in the bit-length of its binary specification.
\end{proof}

One important  class of simplex-like  marginals  is the following.
Let $(\varphi^T,\varphi^T,\varphi^T)$ denote the marginals of the simplex $P_{2r}$, where $r \ge 1$.
Define
\begin{equation}
\label{eq:lattice permutations}
\begin{aligned}
  \lambda^T := \varphi^T + (d_{2r},\dots,d_0), \quad
  \mu^T = \pi^T := \varphi^T + (1^{r+1}, 0^r),
\end{aligned}
\end{equation}
where $d = (d_0,\dots,d_{2r}) \in \mathbb N^{2r+1}$ is such that $\sum_k d_k = r + 1$ and $\sum_k k d_k = r(r+1)$--this   also
implies that $\sum_k k d_{2r-k} = r (r+1)$.
It is not hard to see that these marginals are simplex-like. Indeed,
\begin{align*}
&\quad\sum_{i=0}^{2r} i \lambda^T_i + \sum_{j=0}^{2r} j \mu^T_j + \sum_{k=0}^{2r} k \pi^T_k \\
&= b_{2r} \cdot (1,1,1) +  \sum_{i=0}^{2r} i d_{2r-i} + \sum_{j=0}^r j +
\sum_{k=0}^r k \\
&= b_{2r} \cdot (1,1,1) + 2r (r+1).
\end{align*}
Since $\lvert P_{2r} \rvert \leq n = \lvert P_{2r} \rvert + (r+1) < \lvert P_{2r+1} \rvert$, this is indeed equal to $p(n)$,
where $n$ is the number of boxes of each of $\lambda$, $\mu$, and $\pi$ (compare with \ref{eq:barycenter explicit}).
These marginals arise when embedding permutation matrices  on top of the simplex $P_{2r}$, and in \cite{brunetti} it was shown using this construction that:

\begin{theorem}[\cite{brunetti}]\label{tbrunetti}
The problem of deciding positivity of $t^\lambda_{\mu,\pi}$, given $\lambda,\mu,\pi$ in unary,
 is NP-hard with respect to polynomial-time many-one reductions,  even when $(\lambda^T,\mu^T,\pi^T)$ is restricted to be  of the form
\ref{eq:lattice permutations}.
\end{theorem}

Thus it follows at once from \ref{the:simplex-like}, in conjunction with this result,
that:

\begin{theorem}\label{tkronnp2}
The problem of deciding positivity of the Kronecker coefficient $k^\lambda_{\mu,\pi}$, given $\lambda,\mu,\pi$ in unary, is NP-hard with respect to polynomial-time many-one reductions, even when $(\lambda^T,\mu^T,\pi^T)$ is restricted to be of the form~\ref{eq:lattice permutations}.
\end{theorem}

This proves \ref{tkronnp}. \ref{tkronpos} follows from this result and \ref{the:simplex-like}.

\section{Refined NP-hardness result}\label{srefined}
Let \RestrictedKronecker be the problem of deciding positivity  of $k^\lambda_{\mu,\pi}$, when $\lambda,\mu,\pi$
satisfy the constraints~\ref{tkron1:1}--\short\ref{tkron1:5}.
Specifically:

\begin{definition}\label{dreskron}
Fix a constant $0<\epsilon\leq 1$.
Then \RestrictedKronecker is the problem of deciding, given $m \in \IN$ and partitions $\la,\mu,\pi$ in unary, with $\height(\mu) =
\height(\pi)= m$, whether the Kronecker coefficient $k^\la_{\mu,\pi}$ is positive, assuming that:
\begin{enumerate}[start=2]
\item\label{dreskron:1} $\mu=\pi$,
\item\label{dreskron:2} $\height(\la) \leq m^\epsilon$,
\item\label{dreskron:3} $(\lambda,\mu,\pi)$ is in the Kronecker cone $\kron(m)$,
\item\label{dreskron:4}  $|\lambda|=|\mu|=|\pi| \le  m^3$, and
\item\label{dreskron:5} $\lambda$ is not a hook.
\end{enumerate}
\end{definition}

For the proof of \ref{tkron1}, we need the following
refinement of \ref{tkronnp2}:

\begin{theorem}\label{tnprefined}
\RestrictedKronecker is NP-hard with respect to polynomial-time many-one reductions.
\end{theorem}

We first we need some auxiliary results.

\begin{lemma}\label{lkron}
  Let $\lambda$ and $\mu$ be partitions so that $\height(\lambda) \leq h^2$,
where $h$ is the height of the smallest column of $\mu$.
  Then $(\lambda,\mu,\mu)$ is in the Kronecker cone $\kron(l)$, where $l=\max \{\height(\lambda), \height(\mu)\}$.
\end{lemma}
\begin{proof}
 It is shown in~\cite{burgnonvanish} that $(\lambda,\delta,\delta)$ is in the Kronecker cone whenever $\delta$ is a rectangle of height at least $h$,
and $\height(\lambda) \le h^2$.
  As the Kronecker cone is a cone, this is also true if we rescale each of $\lambda$ and $\delta$ by an arbitrary positive number.

  Let us write $\mu$ as a sum of rectangles $\delta^{(1)} + \dots + \delta^{(k)}$, where our assumption implies that each $\delta^{(j)}$ has height at least $h$.
It is easy to see that  $\lambda$ can be written as a sum $\lambda^{(1)} + \dots + \lambda^{(k)}$, where each $\lambda^{(j)}$ is a rational partition with the same  size  as $\delta^{(j)}$ (i.e., $|\lambda^{(j)}|=|\delta^{(j)}|$),
and with  no more than $h^2$ rows.
  By the preceding argument, each $(\lambda^{(j)},\delta^{(j)},\delta^{(j)})$ is in the Kronecker cone $\kron(l)$.
As cones are closed under addition,  $(\lambda,\mu,\mu)$ is likewise in the Kronecker cone.
\end{proof}

Next, we generalize \ref{the:simplex-like}  to a larger class of marginals.
Let $(\lambda,\mu,\pi)$ be simplex-like, and let $P$ be a corresponding point set with marginals $(\lambda^T,\mu^T,\pi^T)$, so that $P_r \subseteq P \subsetneq P_{r+1}$ for some $r$ (\ref{lsimpyr}).
Let $Q$ denote the following point set obtained by adjoining  $P$ to  a rectangular box of size  $a \times b \times c$, where $b,c \geq r+1$:
\begin{equation}
\label{eq:simplex-like point set on a pedestal}
\begin{aligned}
  Q &= \{0,\dots,a-1\} \times \{0,\dots,b-1\} \times \{0,\dots,c-1\} \\
  &\cup \{ (a+x,y,z) : (x,y,z) \in P \}.
\end{aligned}
\end{equation}
Then the marginals of $Q$ are given by $(\tilde\lambda^T,\tilde\mu^T,\tilde\pi^T)$, where
\[ \tilde \lambda^T  =((bc)^a,\lambda^T), \quad \tilde \mu^T = \mu^T + ((ac)^b), \quad
\tilde \pi^T = \pi^T + ((ab)^c), \]
so that
\begin{equation}
\label{eq:simplex-like on a pedestal}
  \tilde\lambda = (a^{bc}) + \lambda, \quad
  \tilde\mu = (b^{ac}, \mu), \quad
  \tilde\pi = (c^{ab}, \pi).
\end{equation}

\begin{definition}
We call $(\tilde\lambda,\tilde\mu,\tilde\pi)$ \emph{pedestalled-simplex-like} if it is of the form \ref{eq:simplex-like on a pedestal} for some simplex-like $(\lambda,\mu,\pi)$ with at most $r+1$ columns each, $a\geq 0$, and $b,c \geq r+1$.
\end{definition}

Then we have the following generalization  of \ref{lsimpyr}.

\begin{lemma}
\label{lem:pedestalled-simplex-like}
  Let $(\tilde\lambda,\tilde\mu,\tilde\pi)$ be pedestalled-simplex-like, i.e., of the form~\ref{eq:simplex-like on a pedestal}
 for some simplex-like $(\lambda,\mu,\pi)$.
  Then~\ref{eq:simplex-like point set on a pedestal} defines a bijection between point sets $P$ with marginals $(\lambda^T,\mu^T,\pi^T)$ and point sets $Q$ with marginals $(\tilde\lambda^T,\tilde\mu^T,\tilde\pi^T)$.
  In particular, any such point set $Q$ is a pyramid.
\end{lemma}
\begin{proof}
  It suffices to show that any point set $Q$ with marginals $(\tilde\lambda,\tilde\mu,\tilde\pi)$ is of the form~\ref{eq:simplex-like point set on a pedestal}.
  For this, observe that according to the definition of $\tilde\lambda$, the first $a$ x-slices of $Q$ contain exactly $bc$ points each.
  On the other hand, the definition of $\tilde\mu$ and $\tilde\pi$ implies that there are at most $b$ non-zero y-slices and at most $c$ non-zero z-slices, so that $Q$ is a point set in $\mathbb N \times \{0,\dots,b-1\} \times \{0,\dots,c-1\}$.
  It follows that the first $a$ x-slices must each be filled completely without holes by rectangles of size $b \times c$.
  Therefore we may write $Q$ in the form~\ref{eq:simplex-like point set on a pedestal},
and it is clear that $P$ has the correct marginals $(\lambda^T,\mu^T,\pi^T)$.
Since $(\lambda,\mu,\pi)$ is simplex-like, $P_r \subseteq P \subseteq P_{r+1}$ (\ref{lsimpyr}).
Finally, observe that any $Q$ of the form~\ref{eq:simplex-like point set on a pedestal} is clearly a pyramid.
\end{proof}

The following result generalizes \ref{the:simplex-like}.

\begin{theorem}
\label{theorem:pedestalled-simplex-like}
  Let $(\tilde\lambda,\tilde\mu,\tilde\pi)$ be pedestalled-simplex-like, i.e., of the form~\ref{eq:simplex-like on a pedestal}
for some simplex-like $(\lambda,\mu,\pi)$.
  Then:
  \[ k^{\tilde\lambda}_{\tilde\mu,\tilde\pi} = t^{\tilde\lambda}_{\tilde\mu,\tilde\pi} = t^{\lambda}_{\mu,\pi} = k^{\lambda}_{\mu,\pi} \]
\end{theorem}
\begin{proof}
  The first equality follows from \ref{cor:upper and lower match}, as according to \ref{lem:pedestalled-simplex-like} any point set with marginals $(\tilde\lambda^T,\tilde\mu^T,\tilde\pi^T)$ is necessarily a pyramid.
  The middle equality follows  from \ref{lem:pedestalled-simplex-like}.
  The last equality is  \ref{the:simplex-like}.
\end{proof}


We can now give the proof of \ref{tnprefined}.

\begin{proof}[\ref{tnprefined}]
We  apply the pedestal construction to marginals of the form~\ref{eq:lattice permutations}.
Let us choose a rectangular box  of size $c \times s \times s$, where $s = 2r+1$.
That is, we set
\begin{equation}
\label{eq:lattice permutations on a pedestal}
\begin{aligned}
  \lambda^T &= ((s^2)^c, (\varphi^T + (d_{2r},\dots,d_0))), \\
  \mu^T = \pi^T &= ((cs)^s) + \varphi^T + (1^{r+1}, 0^r),
\end{aligned}
\end{equation}
for some $d = (d_0,\dots,d_{2r}) \in \mathbb N^{2r+1}$ such that $\sum_k d_k = r + 1$ and $\sum_k k d_k = r(r+1)$,
where $\varphi^T$ denotes the marginal of the simplex $P_{2r} = P_{s-1}$.
Furthermore, we set
 $c := \lceil s^{2/\epsilon - 1} \rceil$.

The problem of deciding  positivity of $k^\lambda_{\mu,\pi}$, with $(\lambda,\mu,\pi)$ restricted as above,
is NP-hard, since by  \ref{theorem:pedestalled-simplex-like}, these  Kronecker coefficients agree with the  ones
in \ref{tkronnp2}, and  we can  transform instances of the latter to instances of the former
in polynomial time  (as $\epsilon$ is fixed).

We now verify that the five constraints in \ref{dreskron}
are all satisfied. The first is clearly satisfied. For the second,
  \[ \height(\lambda) = s^2 \leq (cs)^\epsilon \leq \height(\mu)^\epsilon = m^\epsilon, \]
  by our choice of $c = c(s)$.
The third follows from \ref{lkron}, as $\height(\lambda) = s^2$, while every column in $\mu$ is of height at least $cs \geq s$.
The fourth follows,   since
\begin{align*}
&|\lambda|= c s^2 + |\varphi^T| + \sum_k d_k =
c s^2 + \frac{(s-1)s(s+1)}6 + r+1 \\
&\le ( c s)^3 \le   \height(\mu)^3 =   m^3,
\end{align*}
assuming that $r$ is large enough.
Finally, it is clear that  $\lambda$ is not a hook.
\end{proof}

\section{Construction of  vanishing Kronecker coefficients with  partition triples in the Kronecker cone}\label{scone}
In this section, we prove \ref{tkron1} and \ref{tkron11}.

\subsection{Proof of \ref{tkron1}}\label{sproofkron1}
By \cite{fortune}, if there exists a co-sparse NP-complete language under polynomial-time many-one-reductions, then P=NP.
\ref{tnprefined}, in conjunction with this result, implies that, assuming P$\neq$NP, the set of partitions triples satisfying the constraints \ref{tkron1:0}--\short\ref{tkron1:5} is non-sparse, i.e, its cardinality is superpolynomial in $m$.
(The result in~\cite{fortune} applies to NP-complete sets, rather than NP-hard sets. But we can still apply this result
to the NP-complete set of simplex-like partition triples (cf.~\ref{tkronnp2}) with positive Kronecker coefficients to get the desired conclusion.)
To prove \ref{tkron1}, we have to get rid of the assumption that P$\neq$NP and replace the superpolynomial bound by $\Omega(2^{m^a})$  bound, for some positive constant $a$.
This will be achieved by \ref{lnonsparse} and \ref{tkron1crux} below.

Recall from \cite{johnson} that the \ThreeDMatching problem is to decide, given a set $M \subseteq W \times X \times Y$, where
$W$, $X$, $Y$ are disjoint sets of size $q$, whether $M$ contains a \emph{(perfect) matching}, i.e., a subset $M' \subseteq M$ of size $q$ such that no two elements of $M'$ agree in any coordinate.
Without loss of generality, we assume henceforth that each element in $W\cup X \cup Y$ appears in some triple of $M$.
We denote instances of \ThreeDMatching by tuples $(M, W, X, Y, q)$.
It is known that the \ThreeDMatching problem is NP-complete \citep{johnson}.

\begin{lemma}\label{lnonsparse}
The number of instances $(M,W,X,Y,q)$ of the \ThreeDMatching problem with total bit-length $\le N$ such that
$M$ does not have a matching is  $\Omega(2^{N^b})$ for some positive
constant $b < 1$.

Furthermore, such instances can be constructed explicitly.
That is, for some positive constant $b<1$,
there is a polynomial-time-computable one-to-one function that maps any
pair of the form $(N,\sigma)$, where $N$ is a positive integer and
$\sigma$ is a binary string of length $\le N^b$, to an
instance of \ThreeDMatching problem without matching of bitlength $\le N$.
\end{lemma}
\begin{proof}
Consider any fixed instance $(M_0,W_0,X_0,Y_0,q_0)$ of \ThreeDMatching, such that $M_0$ does not have a matching.
Its bitlength is thus  a constant.
Given any instance $(M,W,X,Y,q)$ of \ThreeDMatching, with
$W,X$ and $Y$  disjoint from $W_0,X_0$ and $Y_0$, consider the  padded
instance $(M\cup M_0, W \cup W_0, X \cup X_0, Y \cup Y_0, q+q_0)$.
Clearly, $M\cup M_0$ also does not have a matching. The number of instances of
the form $(M\cup M_0, W \cup W_0, X \cup X_0, Y \cup Y_0, q+q_0)$ with bitlength $\le N$ is clearly $\Omega(2^{N^b})$ for some
positive constant $b<1$. This is because, for a given $q$,  the total number of
instances of the form $(M,W,X,Y,q)$ is $2^{q^3}$, and the bit-length
of the specification of any instance of this form is $O(q^3)$. (We assume that
$M$ is specified by its $q \times q \times q$ adjacency matrix.)
Furthermore, it is easy to show that  the  padded instances of
the form $(M\cup M_0, W \cup W_0, X \cup X_0, Y \cup Y_0, q+q_0)$
can be constructed explicitly.
\end{proof}

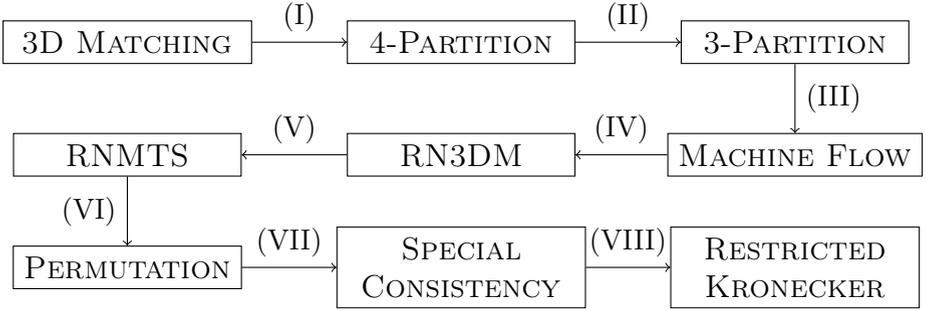
\begin{figure}
\centering
\begin{tikzpicture}[scale=0.5,auto,block/.style={rectangle, draw=black, align=center, minimum width=3cm}]
\draw (0,0) node[block] (3DM) {\parbox{3cm}{\centering\ThreeDMatching}};
\draw (8.8,0) node[block] (4P) {\FourPartition};
\draw (17.6,0) node[block] (3P) {\ThreePartition};
\draw (0,-3) node[block] (RNMTS) {\RNMTS};
\draw (8.8,-3) node[block] (RN3DM) {\RNThreeDM};
\draw (17.6,-3) node[block] (MF) {\MachineFlow};
\draw (0,-6) node[block] (Permutation) {\Permutation};
\draw (8.8,-6) node[block] (SPEC) {\parbox{3cm}{\centering\SpecialConsistency}};
\draw (17.6,-6) node[block] (RK) {\parbox{3cm}{\centering\RestrictedKronecker}};
\draw[->] (3DM.east) --node[above,midway] {\small\ref{tkron1crux:i}} (4P.west);
\draw[->] (4P.east) --node[above,midway] {\small\ref{tkron1crux:ii}} (3P.west);
\draw[->] (3P.south) --node[right,midway] {\small\ref{tkron1crux:iii}} (MF.north);
\draw[->] (MF.west) --node[above,midway] {\small\ref{tkron1crux:iv}} (RN3DM.east);
\draw[->] (RN3DM.west) --node[above,midway] {\small\ref{tkron1crux:v}} (RNMTS.east);
\draw[->] (RNMTS.south) --node[left,midway] {\small\ref{tkron1crux:vi}} (Permutation.north);
\draw[->] (Permutation.east) --node[above,midway] {\small\ref{tkron1crux:vii}} (SPEC.west);
\draw[->] (SPEC.east) --node[above,midway] {\small\ref{tkron1crux:viii}} (RK.west);
\end{tikzpicture}
\caption{Sequence of reductions used in the proof of \ref{tkron1crux}.}\label{fig:reductions}
\end{figure}

Recall (cf.~\ref{sprooftech}) that a polynomial-time one-one reduction means an injective polynomial-time many-one reduction.

\begin{theorem}\label{tkron1crux}
There exists a polynomial-time one-one reduction $\phi$ from the set of instances $(M,W,X,Y,q)$ of \ThreeDMatching  of total bit-length $n$ to the set of partition triples  $(\lambda,\mu,\pi)$ satisfying the conditions \ref{tkron1:1}--\short\ref{tkron1:5} (i.e., instances of \RestrictedKronecker), with $m=\poly(n)$, such that $M$ contains a matching iff the Kronecker coefficient associated with the partition triple $\phi(E)$ is positive.
\end{theorem}
\begin{proof}
Since \ThreeDMatching is in NP, it follows from \ref{tnprefined} that there exists a polynomial-time many-one reduction $\phi$ from \ThreeDMatching to the \RestrictedKronecker  problem of
deciding positivity of the Kronecker coefficient $k^\lambda_{\mu,\pi}$, with $(\lambda,\mu,\pi)$ satisfying the constraints \ref{tkron1:1}--\short\ref{tkron1:5}.
We have to show that this reduction $\phi$ can be chosen to be injective.
We can obtain such an injective reduction $\phi$ by composing the following sequence of polynomial-time computable one-one-reductions (\ref{fig:reductions}):
\begin{enumerate}[label=(\Roman*)]
\item\label{tkron1crux:i} From \ThreeDMatching to \FourPartition (cf.~Theorem 4.3 in \citealp{johnson}):

The \FourPartition problem is to decide, given a set $A$ of size $4 m$, a positive integer bound $B$, a positive integer size $s(a)$ for each $a \in A$ such that $B/5 < s(a) < B/3$ and $\sum_{a \in A} s(a)= m B$, whether $A$ can be partitioned into $m$ disjoint subsets $A_1,\ldots,A_m$, each of size four,  such that, for each $1 \le i \le m$, $\sum_{a \in A_i} s(a)=B$.
We denote such an instance of \FourPartition  by the tuple $(A,m,B,s)$.

The reduction in \cite{johnson}  maps a given instance $(M,W,X,Y,q)$ of
\ThreeDMatching to an instance $(A,m,B,s)$ of 4-partition, where:
\begin{itemize}
\item The set $A$ has $4 |M| = O(q^3)$ elements, one for each occurrence of a member of $W \cup X \cup Y$ in a triple in $M$ and one for each triple in $M$.
\item Let $W=\{w_1,\ldots,w_q\}$, $X=\{x_1,\ldots,x_q\}$, and $Y=\{y_1,\ldots,y_q\}$. Given any $z \in W \cup X \cup Y$,
let $N(z)$ denote the number of triples in $M$ that contain $z$, and  let $z[1], z[2],  \ldots, z[N(z)]$ denote the elements
in $A$ corresponding to $z$. Let $r=32 q$, and define
\begin{align*}
s(w_i[1]) &= 10 r^4 + i r +1, && 1 \le i \le q, \\
s(w_i[l]) &= 11 r^4 + i r + 1, && 1 \le i \le q, 2 \le l \le N(w_i), \\
s(x_j[1]) &= 10 r^4 + j r^2 + 2, && 1 \le j \le q, \\
s(x_j[l]) &= 11 r^4 + j r^2 + 2, && 1 \le j \le q, 2 \le l \le N(x_j), \\
s(y_k[1]) &= 10 r^4 + k r^3 + 4, && 1 \le k \le q, \\
s(y_k[l]) &= 8 r^4 + k r^3 + 4, && 1 \le  k \le q, 2 \le l \le N(y_k).
\end{align*}
\item Let $u_l$ denote the single element corresponding to a particular triple $m_l=(w_i,x_j,y_k)\in M$.
For any such $u_l$, let $s(u_l)= 10 r^4 - k r^3 - j r^2 - i r + 8$.
\item $B = 40 r^4 + 15$.
\end{itemize}
Note that $\max\{s(a)| a \in A\} \le 2^{16} |A|^4$.
This means \FourPartition is NP-complete in the strong sense~\cite{johnson}.
It can be checked that this reduction is injective.

\item\label{tkron1crux:ii} From \FourPartition to  \ThreePartition (cf.~Theorem 4.4  in \cite{johnson}):

The \ThreePartition problem is to decide, given a set $A$ of size $3 m$, a positive integer bound $B$, a positive integer size $s(a)$ for each $a \in A$ such that $B/4 < s(a) < B/2$ and $\sum_{a \in A} s(a)= m B$, whether $A$ can be partitioned into $m$ disjoint subsets $A_1,\ldots,A_m$, each of size three, such that, for each $1 \le i \le m$, $\sum_{a \in A_i} s(a)=B$.
We denote an instance of \ThreePartition by the tuple $(A, m, B, s)$.

The reduction in \cite{johnson} maps an instance $(A,m,B,s)$ of \FourPartition, with $|A|=4 m$ and $\max\{s(a)| a \in A\} \le 2^{16} |A|^4$, to the instance $(A',m',B',s')$ of \ThreePartition, where $A'$ has  $m'=O(m^2)$ elements:
one element $w_i$ for each element $a_i$  of $A$,
two elements $u_{i,j}$ and $\bar u_{i,j}$ for each pair $(a_i,a_j)$ of elements from $A$, and
$8 m^2 - 3 m$ filler elements $u^*_k$, $1 \le k \le 8 m^2-3 m$.
Their sizes are:
\begin{align*}
s'(w_i) &= 4 ( 5 B + s(a_i)) + 1, \\
s'(u_{i,j})&= 4 (6 B - s(a_i) - s(a_j)) + 2, \\
s'(\bar u_{i,j}) &= 4 ( 5 B + s(a_i) + s(a_j)) + 2, \\
s'(u^*_k) &= 20 B.
\end{align*}
We let $B'= 64 B +4$.
It can be checked that this reduction is injective.

\item\label{tkron1crux:iii} From \ThreePartition to \MachineFlow,
the decision version of the two-machine flow scheduling problem with unit processing times
defined in Chapter 3 in \cite{yu} (where it is called \textsc{F2UD'}):

The \MachineFlow problem is to decide, given two machines M1 and M2, each of which can process at most one job at a time, and $n$ jobs $j$, $1 \le j \le n$, where each job takes unit processing time and the job $j$ is assigned a delay $l_j$ that describes the minimum amount of time between the completion of the job $j$ on M1 and its start on M2, and a threshold $y$, whether there exists a feasible schedule of the jobs so that the last job is completed before time $y$.

The reduction in \cite{yu} from \ThreePartition to \MachineFlow goes as follows.
Without loss of generality, we consider a modified version of \ThreePartition by multiplying the partition elements by $4m$.
Thus we are given a set of positive integers $A=\{ a_1,\ldots, a_{3m}\}$  and a positive integer $B$ such that
(1) $B < a_i < 2 B$ for all $i$, (2) $\sum_j a_j = 4 m B$, (3) $a_i = 0$ (mod $m$) for all $i$, and (4) $4 B=0$ (mod $m$).
The problem is to decide if $A$ can be partitioned into $m$ disjoint 3-element subsets $A_1,\ldots, A_m$ such that $\sum_{a_j \in A_i} a_j =
4 B$, for all $i$.
An instance of this modified version of \ThreePartition is mapped to an instance of \MachineFlow with delays
(1) $l_j = a_j$ for $1 \le j \le 3 m$, (2) $l_j=0$ for $ 3 m + 1 \le j \le 4 m B$, (3) $l_j= u  + 1$ for $4 m B + 1 \le j \le m u$, where $u= 4 (m+1) B$, and (4) the threshold $y= n + 4 m B +2$, where $n= m u$ is the total number of jobs.
It can be checked that this reduction is injective.

\item\label{tkron1crux:iv} From \MachineFlow  to  \RNThreeDM (Restricted Numerical 3-Dimensional Matching, cf.~page 31 in \citealp{yu}):

The \RNThreeDM problem is to decide, given a positive integer set $U=\{u_1,\ldots,u_n\}$ and a positive integer $e$ such that $\sum_{j=1}^n u_j + n(n+1) = n e$, whether there exist two $n$-permutations $\lambda$ and $\mu$  such that $j + \lambda(j) + u_{\mu(j)} = e$ for $1 \le j \le n$. (It can be assumed that each $u_i< e-1$).

The reduction (Corollary 3 on page 32 in \citealp{yu}) maps an instance of \MachineFlow  to that of \RNThreeDM given by $u_j=l_j$ for $1 \le j \le n$, and $e=y$.
We assume that the instance of \MachineFlow here  arises in the reduction from \ThreePartition to \MachineFlow given in~\ref{tkron1crux:iii} above.
This will ensure that $\sum_j u_j + n(n+1) = ne$ and  each  $u_i < e-1$, which we require for~\ref{tkron1crux:v} below to be injective.
It can be checked that this reduction is injective.

\item\label{tkron1crux:v} From \RNThreeDM to \RNMTS (Restricted Numerical Matching with Target Sums, cf.~\citealp{brunettibinary}):

The \RNMTS problem is to decide, given positive integers $y_1,\ldots, y_n$ such that  $2 \le y_1 \le y_2 \le \cdots \le y_n \le 2 n$ and $\sum_i y_i = n(n+1)$, if there exist $n$-permutations $\sigma$ and $\pi$ such that $\sigma(k) + \pi(k)= y_k$ for $1 \le k \le n$.

\RNThreeDM is mapped to \RNMTS by letting $y_j= e - u_j$ and then reordering the $y_j$ as per their values.
It can be checked that this reduction is injective.

\item\label{tkron1crux:vi} From \RNMTS to  \Permutation:

The \Permutation problem (page 69 in \citealp{brunettibinary}, where it is called \Permutation($S_3$)) is to decide, given non-negative integers $z_2,\ldots, z_{ 2 n} \in \{0,\ldots,n\}$, whether there exists an $n \times n$ permutation matrix $P$ such that $\sum_{i,j: i+j = l} P_{i,j}= z_l$ for $2 \le l \le 2 n$.

The reduction in \cite{brunettibinary} maps an instance $y=(y_1,\ldots,y_n)$ of \RNMTS  to an instance $z=(z_2,\ldots,z_{2 n})$  of \Permutation by setting $z_l=|\{ k \le n \ | \ y_k = l\}|$.
It can be checked that this reduction is injective.

\item\label{tkron1crux:vii} From \Permutation to \SpecialConsistency:

\SpecialConsistency is the problem addressed in \ref{tbrunetti}, namely, the problem of deciding positivity of $t^\lambda_{\mu,\pi}$, given $\lambda,\mu$ and $\pi$ in unary, when $(\lambda^T,\mu^T,\pi^T)$ is restricted to be  of the form~\ref{eq:lattice permutations}.

The reduction in \cite{brunetti} maps an instance $z=(z_2,\ldots,z_{2n})$ of \Permutation to an instance  $(\lambda,\mu,\pi)$ of \SpecialConsistency satisfying~\ref{eq:lattice permutations}, with $r=n-1$ and $d_i=z_{i+2}$, $0 \le i \le 2 r$ (it can be shown that $\sum_k d_k = r+1$, and $\sum_k k d_k = r (r+1)$).
This reduction is injective.

\item\label{tkron1crux:viii} From \SpecialConsistency to  \RestrictedKronecker:

The reduction given in the proof of \ref{tnprefined} is also injective.
\end{enumerate}
\end{proof}

\ref{tkron1} follows from \ref{lnonsparse} and \ref{tkron1crux}.
This proof also shows that the superpolynomially many partition triples
in \ref{tkron1} can be constructed explicitly (as defined in \ref{suncond}).

\begin{remark}
In the preceding proof, we can use, in place of \ThreeDMatching,  any problem  in NP which has a polynomial-time-computable padding function \citep{hartmanis}, and which can be reduced  by a polynomial-time one-one reduction to  \RestrictedKronecker.
For example, \textsc{Sat} also has a polynomial-time-computable  padding function, and it  can also be reduced  by a polynomial-time one-one reduction to \RestrictedKronecker.
This reduction is obtained by composing the injective reduction from \textsc{Sat} to \ThreeDMatching given in \cite{johnson} with the injective reduction from \ThreeDMatching to  \RestrictedKronecker given in the proof of \ref{tkron1crux}.
\end{remark}

\begin{example}\label{sexample}
Though the reduction $\phi$ in \ref{tkron1crux} is po\-ly\-no\-mi\-al-time computable, the blow-up in size can be substantial.
For example, let us start with a trivial instance of the \ThreeDMatching problem, wherein $q=2$, $W=\{w_1,w_2\}$, $X=\{x_1,x_2\}$, $Y=\{y_1,y_2\}$, and $M=\{(w_1,x_1,y_1),(w_2,x_1,y_2), (w_1, x_2, y_2)\}$.
Clearly, $M$ does not contain a matching.

It can be checked that $\phi(M,W,X,Y,q)$, with $\epsilon=1$ in condition \ref{tkron1:2}, is a partition triple whose height is $> 10^{16}$ and the total size is $> 10^{46}$.
By \ref{tkron1crux}, the Kronecker coefficient associated with this partition triple is zero.
One cannot verify this fact directly using a computer, since computation of Kronecker coefficients for partition triples of this height and size
is far beyond the reach of  computer algebra systems.
Thus \ref{tkron1crux} maps instances of \ThreeDMatching which do not contain matching for trivial reasons to partition triples whose associated Kronecker coefficients vanish for highly nontrivial reasons.
Thus the image of something trivial is highly nontrivial.
This happens because of the nontriviality of the sequence of reductions that produce the image.
\end{example}

\subsection{Proof of \ref{tkron11}}
For the proof of \ref{tkron11}, we need the following lemma, which proves a variant of the result in \cite{fortune} that coNP-complete languages cannot be sparse unless P=NP.

\begin{lemma}\label{lsparse}
  Let $\mathcal L$ be a $\mbox{coNP}$-hard language given as a disjoint union
  \[ \mathcal L = \mathcal L' \cup \mathcal L'', \]
  where $\mathcal L'$ is sparse (i.e., there are only $\poly(n)$  words of length $n$ in $\mathcal L'$) and $\mathcal L'' \in \mbox{NP} \cap \mbox{coNP}$.
  Then $\mbox{coNP} = \mbox{NP}$.
\end{lemma}
\begin{proof}
  We will show that the assumptions imply that $\mbox{SAT}^c$ (the complement of SAT) is in $\mbox{NP}$ -- this would imply that $\mbox{coNP} \subseteq  \mbox{NP}$, and hence, $\mbox{coNP} = \mbox{NP}$.
  For this, we adapt the proof in \cite{mahaney,fortune}.

  Since $\mathcal L$ is $\mbox{coNP}$-hard, there exists a polynomial-time many-one
 reduction $R$ such that $R(\mbox{SAT}) \subseteq {\mathcal L}^c$ and $R(\mbox{SAT}^c) \subseteq \mathcal L$.
Since $\mathcal L''$ is in $\mbox{NP} \cap \mbox{coNP}$, there exist non-deterministic Turing machines $M_1$ and $M_2$ such that, given input $x$, $M_1$
halts (in polynomial time)  if and only if $x \in \mathcal L''$, while $M_2$ halts (in polynomial time)
if and only if $x \not\in \mathcal L''$.

  Let $F$ be a formula for which we have to decide unsatisfiability.
  We perform depth-first search on  the binary tree obtained by self-reducing $F$ (the root of this tree is  $F$,
and the  children of a node $G$ are $G_0$ and $G_1$, the formulas of smaller size obtained by specializing
the first variable in $G$  to $\mbox{true}$ or $\mbox{false}$, and applying trivial simplifications), starting at the root node.
  We maintain a table $\mathcal U$ of labels ($R$-values) of unsatisfiable formulae,
starting with $\mathcal U := \{R(\mbox{false})\}$.
  At each node $G$, we first compute $R(G)$ and then do one of the following:
  \begin{enumerate}
    \item If $R(G) \in \mathcal U$, prune the subtree and return to the parent node.
    \item Otherwise, if $G = \mbox{true}$, enter an infinite loop.
    \item Otherwise, run both non-deterministic Turing machines $M_1$ and $M_2$ in parallel on the input $R(G)$ until one of the two halts (which will always happen, for some sequence of non-deterministic choices,  in polynomial time):
      \begin{enumerate}
        \item If $M_1$ halts (in which case $R(G) \in \mathcal L'' \subseteq \mathcal L$, and hence, $G$ is unsatisfiable),
add $R(G)$ to $\mathcal U$, prune the subtree and return to the parent node.
        \item If $M_2$ halts, visit both children $G_0$ and $G_1$. Upon return (if this happens),
it will always be true that $G_0$ and $G_1$ are unsatisfiable, and hence $R(G_0), R(G_1) \in \mathcal U$
and  $G$ is unsatisfiable. Thus add $R(G)$ to $\mathcal U$ and return to the parent node.
      \end{enumerate}
  \end{enumerate}
  It is clear that this algorithm can be understood as a non-deterministic Turing machine that halts if and only if $F$ is unsatisfiable.

It suffices to show that, if $F$ is  unsatisfiable, this algorithm  halts in polynomial time.
  For this, it suffices to show that the number of \emph{interior} nodes that are visited by the algorithm is polynomial in the size $|F|$ of the formula $F$ (since the tree is binary, the number of visited leaves is at most twice the number of visited interior nodes).
  Now observe that interior nodes only arise in the case where $M_2$ halts on input $R(G)$, in which case $R(G) \in \mathcal L'$.
  Thus any interior node is necessarily labeled by an element of the sparse set $\mathcal L'$.
  We can thus conclude the argument precisely as in Lemma 2.2 of \cite{mahaney}:
  If $G$ and $G'$ are two interior nodes that have the same label, $R(G) = R(G') \in \mathcal L'$, then they necessarily ought to appear in the same branch of the search tree (because we proceed by depth-first search).
  As the depth of the tree is no more than $m$ -- the number of variables in $F$ -- we find that each label can occur at most $m$ times.
  Therefore, the number of visited interior nodes can be upper bounded by $m \cdot p(q(\lvert F \rvert))$, where $q$ is a polynomial that bounds the increase in length induced by the reduction $R$ and
 $p=p(n)$ is a polynomial that bounds the number of strings of  length $\leq n$ in the sparse set $\mathcal L'$.
  We conclude that $\mbox{SAT}^c \in \mbox{NP}$.
\end{proof}

Another ingredient needed for the proof of \ref{tkron11} is the following result.

\begin{theorem}[\cite{BCMW}]\label{tconp}
The problem of deciding if $(\lambda,\mu,\pi) \in \kron(m)$ is in $\mbox{NP}
\cap \mbox{coNP}$. Here, $m$ denotes the maximum height of $\lambda$, $\mu$, or $\pi$,
and  the partition triple $(\lambda,\mu,\pi)$ is given in unary.
\end{theorem}

We can now give the proof of \ref{tkron11}:

\begin{proof}[\ref{tkron11}]
For given $m$, let $\mathcal L$ be the set of partition triples $(\lambda,\mu,\mu)$ satisfying the constraints \ref{tkron1:0}--\short\ref{tkron1:5}.
Let $\mathcal L'$ be the set of partition triples $(\lambda,\mu,\mu)$  satisfying both the constraints \ref{tkron1:0}--\short\ref{tkron1:5} and \ref{tkron11:6}--\short\ref{tkron11:7}.
Let $\mathcal L''$ be the set of partition triples satisfying the constraints \ref{tkron1:0}--\short\ref{tkron1:5} such that either~\short\ref{tkron11:6} or \short\ref{tkron11:7} in \ref{tkron11} are violated.
Then, clearly, $\mathcal L=\mathcal L' \cup \mathcal L''$.

In the definition of $\mathcal L''$, we can drop the constraint~\short\ref{tkron1:0}, since it is automatically satisfied if \short\ref{tkron11:6} or \short\ref{tkron11:7} are violated (as $k^\lambda_{\mu,\mu}=k^{\lambda^t}_{\mu^t,\mu} = k^\lambda_{\mu^t,\mu^t}$).
By \ref{tconp}, the problem of deciding whether a partition triple belongs to the Kronecker cone is in $\mbox{NP} \cap \mbox{coNP}$.
It follows that $\mathcal L'' \in \mbox{NP} \cap \mbox{coNP}$.
By \ref{tnprefined}, $\mathcal L$ is coNP-hard.
It now follows from \ref{lsparse} that $\mathcal L'$ is not sparse, assuming  $\mbox{coNP} \not = \mbox{NP}$.
This proves \ref{tkron11}.
\end{proof}

 \section{\texorpdfstring{Correlation between the complexity of~$k^\lambda_{\mu,\pi}$ and~$t^\lambda_{\mu,\pi}$}{Correlation between the complexities of k and t}}\label{srectangular}

There seems to be a surprising correlation between the complexities of $k^\lambda_{\mu,\pi}$ and
$t^\lambda_{\mu,\pi}$.
On the one hand, positivity of $k^\lambda_{\mu,\pi}$ is, in general, NP-hard to decide (\ref{tkronnp2}), just as it is for $t^\lambda_{\mu,\pi}$ (\ref{tbrunetti}).
On the other hand, suppose  $\Pi$ is a subclass of partition triples such that the problem of deciding positivity of $t^\lambda_{\mu,\pi}$, for  $(\lambda,\mu,\pi) \in \Pi$,  is in $P$.
While the corresponding problem of deciding positivity of $k^\lambda_{\mu,\pi}$, for  $(\lambda,\mu,\pi) \in \Pi$, may not always be in $P$, the results in this section suggest that it  may indeed be so  for many ``natural'' subclasses $\Pi$.
In particular, \ref{tpositivet} proved in this section suggests that the problem of deciding positivity of $k^\lambda_{\mu,\pi}$, when
$\mu$ and $\pi$ are rectangular ($=\delta(\lambda)$), is in $P$, as conjectured in~\cite{GCT6}.

\subsection{\texorpdfstring{Interpretation of the $t$-function in terms of hypergraphs}{Interpretation of the t-function in terms of hypergraphs}}
We begin with a lemma that is needed for proving these results.

\begin{definition}
Let $d \in \N$.
An \emph{obstruction predesign} is a hypergraph with $d$ indistinguishable vertices
and with hyperedges that come in three colors, such that every vertex lies in exactly one hyperedge of each color.
The set all hyperedges that have the same color is called a \emph{layer of hyperedges}, so in our case we have three layers of hyperedges.

Let $\lambda,\mu,\pi$ be partitions of $d$.
An obstruction predesign is defined to have \emph{type} $(\lambda,\mu,\pi)$ if
the number of columns in~$\lambda$ of length~$k$ equals the number of hyperedges in layer 1 with $k$ vertices, and
the number of columns in~$\mu$ of length~$k$ equals the number of hyperedges in layer 2 with $k$ vertices, and
the number of columns in~$\pi$ of length~$k$ equals the number of hyperedges in layer 3 with $k$ vertices.

To each vertex we can assign its triple of hyperedges.
An \emph{obstruction design} is an obstruction predesign such that no two vertices have the same triple of hyperedges.
\end{definition}


\begin{lemma}\label{lhyper}
Let $\tilde t^\lambda_{\mu,\pi}$ be the number of obstruction designs of type $(\lambda,\mu,\pi)$.
Then  $t_{\mu,\pi}^\lambda > 0$ iff  $\tilde t_{\mu,\pi}^\lambda >0$.
\end{lemma}
\begin{proof}
Recall from \ref{sbounds} that $t^\lambda_{\mu,\pi}$ is the number of point sets $P \subseteq \mathbb N^3$ with marginals $(\lambda^T, \mu^T, \pi^T)$.
We  define the   $k$-th slice of a point set in direction $i$, $1 \leq i \leq 3$, to be its subset consisting  of the  points that have $k$ as their $i$-th coordinate. (Here the directions $1$, $2$, and $3$ correspond to the $x$, $y$, and $z$ coordinates, respectively.)

We now  prove that $t_{\mu,\pi}^\lambda>0$ iff $\tilde t_{\mu,\pi}^\lambda>0$.
From a point set  $P$ with marginals $(\lambda^T, \mu^T,\pi^T)$,
we can define an obstruction design of type $(\lambda,\mu,\pi)$ by
taking $P$ to be the hypergraph vertex set and making each slice in direction $i$ a hyperedge in layer $i$.
Conversely, from an obstruction design of type $(\lambda,\mu,\pi)$
we can obtain a point set  with marginals  $(\lambda^T,\mu^T,\pi^T)$
as follows:
For each layer we give consecutive numbers (starting at $0$) to each hyperedge, beginning with the largest hyperedge and continuing
in a manner such that the hyperedge sizes form a nonincreasing sequence, i.e., a partition of~$d$.
Since every vertex lies in exactly one hyperedge of each layer, every vertex gives rise to a triple of nonnegative integers,
and the triples are pairwise distinct because no two vertices share all 3 hyperedges.
These triples form a point set with marginals  $(\lambda^T,\mu^T,\pi^T)$.
\end{proof}

Note that $t^\lambda_{\mu,\pi}$ needs not always be equal to $\tilde t^\lambda_{\mu,\pi}$;
this  can happen when the hyperedge sizes are not all distinct.
Though, by the preceding result, the problems of deciding positivity of $t^\lambda_{\mu,\pi}$ and $\tilde t^\lambda_{\mu,\pi}$ are equivalent, obstruction designs introduced  here will turn out to be convenient in the proofs that follow.

\subsection{Littlewood-Richardson coefficients}
Given partitions $\la$, $\mu$, $\pi$ such that $|\la|=|\mu|+|\pi|$, let $\iota := |\la|+\la_1$.
Let $\tilde\la$, $\tilde\mu$, $\tilde\pi$ be the partitions $\la$, $\mu$, $\pi$ with a long first row put on top of their Young diagrams such that $|\tilde\la|=|\tilde\mu|=|\tilde\pi|=3\iota$.
Then, $k^{\tilde \lambda}_{\tilde \mu,\tilde \pi}$ is equal to the Littlewood-Richardson coefficient $c^\lambda_{\mu,\pi}$ associated with the partition triple $(\lambda,\mu,\pi)$ \citep{murnaghan}.

The problem of deciding positivity of $k^{\tilde \lambda}_{\tilde \mu,\tilde \pi}=c^\lambda_{\mu,\pi}$ has a strongly polynomial time algorithm~\citep{knutsontao2,GCT3}.
This is consistent with the following result.

\begin{theorem}\label{tlittle}
For any $\lambda,\mu$, and $\pi$,  $t_{\tilde\mu,\tilde\pi}^{\tilde\la} > 0$. In particular, the problem of deciding
positivity of $t_{\tilde\mu,\tilde\pi}^{\tilde\la}$ is trivial.
\end{theorem}

\begin{proof}
By \ref{lhyper}, it suffices to  construct a hypergraph to show that  $\tilde t_{\tilde\mu,\tilde\pi}^{\tilde\la} > 0$.
We call a hyperedge that contains only a single vertex a \emph{singleton}.
The key property of the constructed hypergraph will be that every vertex lies in 2 singletons and another hyperedge.
By construction $\tilde\la$, $\tilde\mu$, and $\tilde\pi$ each have at least $2\iota$ columns with a single box.
We split the $3\iota$ vertices into three equally sized parts.
The vertices of the first part are contained in singleton hyperedges of layer 2 and layer 3.
The vertices of the second part are contained in singleton hyperedges of layer 1 and layer 3.
The vertices of the third part are contained in singleton hyperedges of layer 1 and layer 2.
The remaining hyperedges of layer $i$ are constructed by freely
partitioning the vertices in the $i$-th part according to the desired
hyperedge sizes.
Since  no two vertices  share all three hyperedges, the theorem is proved.
\end{proof}

\subsection{Partitions of constant height}

{
\newcommand{\mc}{c}

It is known that positivity of $k^\lambda_{\mu,\pi}$ can be decided in polynomial time when $\lambda,\mu$, and $\pi$ have constant
heights. This is consistent with the following result.

\begin{theorem}
Fix a constant $\mc \in \IN$.
If $|\la|=|\mu|=|\pi|$ and
$\height(\la) \leq \mc$, $\height(\mu) \leq \mc$, $\height(\pi) \leq \mc$,
then positivity of $t^\la_{\mu,\pi}>0$ can be decided in polynomial time.
\end{theorem}
\begin{proof}
The algorithm is a hybrid algorithm based on the number of boxes $|\la|$.
The values for $t$ in the case $|\la|< (\mc+2)\mc$ are stored in a database of constant size.
The case $|\la|\geq (\mc+2)\mc$ is trivial, as the following lemma shows.
\end{proof}

\begin{lemma}
If $\height(\la) \leq \mc$, $\height(\mu) \leq \mc$, $\height(\pi) \leq \mc$,
and $|\la|=|\mu|=|\pi|\geq (\mc+2)\mc$, then $t^\la_{\mu,\pi}>0$.
\end{lemma}
\begin{proof}
By \ref{lhyper}, it suffices to construct a hypergraph showing that $\tilde t^\la_{\mu,\pi}>0$.

If $|\la|$ is divisible by $\mc$, then we arrange the vertices in a rectangular array whose columns contain $\mc$ vertices each.
Otherwise we add an extra column containing less than $\mc$ vertices.
The crucial property is that since $|\la|\geq (\mc+2)\mc$ each row contains at least $\mc+2$ vertices.

Proceeding column-wise from top to bottom and from left to right,
we greedily assign vertices to hyperedges of the first layer according to the column lengths of $\mu$.
Note that each hyperedge constructed thus lies either in a single column or in two adjacent columns.
      Likewise, proceeding row-wise from left to right and from top to bottom, we greedily assign vertices to hyperedges of the second layer according to the column lengths of $\pi$.
      Since each row contains at least $c+2$ vertices, a layer 2 hyperedge cannot contain two vertices from the same or adjacent columns.

Consider two vertices $v$ and $w$ that lie in the same layer 1 hyperedge. They either lie in the same column or in adjacent columns.
Therefore no layer 2 hyperedge can contain both $v$ and $w$.
This shows that we can choose an arbitrary layer 3 hyperedge arrangement and see that $\tilde t^\la_{\mu,\pi}>0$.
\end{proof}
}

Note that $t_{\mu,\pi}^\lambda>0$ and $t_{\mu',\pi'}^{\lambda'}>0$ implies $t_{\mu+\mu',\pi+\pi'}^{\lambda+\lambda'}>0$,
and that Lemma 6.5 shows that for constant height the semigroup of triples with positive $t_{\mu,\pi}^\lambda$ is finitely generated, as it is known for $k_{\mu,\pi}^\lambda$ as well.

\subsection{When one partition is a hook}
{
\newcommand{\mk}{k}
\newcommand{\mD}{D}

\cite{blasiak} has a given a $\#P$-formula for $k^\lambda_{\mu,\pi}$ when $\lambda$ is a hook.
The problem of deciding positivity of $k^\lambda_{\mu,\pi}$ in this case may be conjectured to be in $P$ in view of the following result.

\begin{theorem}\label{thook}
Positivity of  $t_{\mu,\pi}^{\la}$, given $\la,\mu$, and $\pi$ in unary,
can be decided in polynomial time
if $\lambda$ is  a hook.
\end{theorem}

We give the proof of \ref{thook} in the remainder of this subsection.
Let $\la=(\mD-\mk+1,1^{\mk-1})$ be a hook partition with $\mD$ boxes. The Young diagram $\la$ has $\mD-\mk$ columns that contain only a single box
and a single column with $\mk$ boxes.
Let $\mu$ and $\pi$ be arbitrary partitions of $\mD$.
By \ref{lhyper}, it suffices to  decide in polynomial time whether $\tilde t_{\mu,\pi}^{\la} > 0$, i.e., if an obstruction design of type $(\la,\mu,\pi)$ exists.

The first key observation is the following.
If we fix all the hyperedges in the $\mu$ and $\pi$ layers and ask whether an obstruction design of type $(\la,\mu,\pi)$ exists
with these prescribed hyperedges, then this is easy to answer:
For fixed set partitions $\alpha$ and $\beta$ that have hyperedge sizes $\mu^T$ and $\pi^T$, respectively,
let $\tilde t_{\mu,\pi}^{\la}(\alpha,\beta)$
denote the number of obstruction designs whose $\mu$ layer is $\alpha$ and whose $\pi$ layer is $\beta$.
We call two vertices that lie in the same $\mu$-hyperedge and the same $\pi$-hyperedge \emph{$(\alpha,\beta)$-equivalent}.
\begin{claim}\label{cla:LRCfirst}
$\tilde t_{\mu,\pi}^{\la}(\alpha,\beta)=0$ iff $\mk$ is larger than the number of $(\alpha,\beta)$-equivalence classes.
\end{claim}
\begin{proof}
Fix $\alpha$ and $\beta$.
If $\mk$ is larger than the number of $(\alpha,\beta)$-equivalence classes, then by the pigeonhole principle the $\la$-hyperedge of size $\mk$ must contain two $(\alpha,\beta)$-equivalent vertices.
Therefore this construction does not yield an obstruction design.
If $\mk$ is smaller than the number of $(\alpha,\beta)$-equivalence classes, the $\la$-hyperedge of size $\mk$ can be chosen to contain pairwise $(\alpha,\beta)$-nonequivalent vertices.
The other vertices are singletons in the $\la$ layer, so this construction yields an obstruction design.
\end{proof}

Thus being able to answer the question whether $\tilde t_{\mu,\pi}^\lambda>0$ is equivalent to answering the following
question: Given hyperedge size vectors $\mu^T$ and $\pi^T$, what is the maximal number of $(\alpha, \beta)$-equivalence
classes, where $\alpha$ and $\beta$ have hyperedge sizes $\mu^T$ and $\pi^T$, respectively?
We will formalize this as a max flow problem with integer edge capacities
(given in unary).
Such a problem can be solved in polynomial time using the Ford-Fulkerson algorithm.
We construct a directed graph with a source vertex, one vertex for each column of $\mu$, one vertex for each column of $\pi$, and one sink vertex.
There are edges from the source vertex to the $\mu$-vertices whose capacity equals the number of boxes in the corresponding column of $\mu$.
Analogously there are edges from the $\pi$-vertices to the sink vertex whose capacity equals the number of boxes in the corresponding column of $\pi$.
Moreover, there is a capacity 1 edge from every $\mu$-vertex to every $\pi$-vertex.
The Ford-Fulkerson algorithm finds in polynomial time an integer solution to the problem of sending the maximum amount of flow through this network with respect to the capacity constraints.
Combined with \ref{cla:LRCfirst} the following claim implies that the Ford-Fulkerson algorithm can be used to decide in polynomial time whether $\tilde t_{\mu,\pi}^\lambda>0$ is positive.
\begin{claim}
A solution with flow at least $\mk$ exists iff there exist $\alpha$ and $\beta$ such that the number of $(\alpha,\beta)$-equivalence classes is at least~$\mk$.
\end{claim}
\begin{proof}
Given $\alpha$ and $\beta$ with at least $\mk$ equivalence classes we construct a solution to the flow problem by sending one flow unit for each equivalence class:
For the equivalence class corresponding to the $i$th $\mu$-column and $j$th $\pi$-column we send a unit from the source vertex to the
$i$th $\mu$-vertex, from there to the $j$th $\pi$-vertex and then to the sink vertex.
This satisfies the capacity constraints and is a solution to the flow problem that sends at least $\mk$ flow units.

From a solution of the max flow problem we readily generate a solution to a relaxed max flow problem where we
remove the capacities on the edges from the
$\mu$-vertices to the $\pi$-vertices. We send flow units on additional arbitrary paths from the source to the sink.
Once all capacities are saturated we are guaranteed to send exactly $\mD$ flow units.
From this new solution we construct set partitions $\alpha$ and $\beta$ by defining that the size of the $(\alpha,\beta)$-equivalence class corresponding to the $i$th $\mu$-column and the $j$th $\pi$-column is the amount of flow from the $i$th $\mu$-vertex to the $j$th $\pi$-vertex.
So if the original solution had at least $\mk$ flow units, then there are at least $\mk$ $(\alpha,\beta)$-equivalence classes in our construction.
\end{proof}
}

\subsection{Rectangular Kronecker coefficients}
It is conjectured in \cite{GCT6} that the problem deciding positivity of the rectangular Kronecker coefficient
$k^\lambda_{\delta(\lambda),\delta(\lambda)}$ is in $P$. This  is supported by the following
result.

\begin{theorem}\label{tpositivet}
Let $\lambda$ be any partition with $d r$ boxes and at most $\min(d^2,r^2)$ rows, and let $\delta = \delta(\lambda) = (d,\ldots,d)$ ($r$ times).
Then
$t_{\delta,\delta}^\lambda>0$. In particular, the problem of deciding positivity of $t_{\delta,\delta}^\lambda$ is trivial.
\end{theorem}
Since $k_{\delta,\delta}^\lambda = 0$ if $\height(\lambda) > min(d^2,r^2)$, the constraint on $\lambda$ here
is  very natural.

\begin{proof}
By \ref{lhyper}, it suffices to  prove positivity of $\tilde t_{\delta,\delta}^\lambda$.
We do this  by an explicit construction.

The case $d \geq r$ is easier, so we handle this case first.
We have to construct an obstruction design with $dr$ vertices and go about it as follows.
Let $i \mathop{\text{rem}} d$ denote the remainder when dividing $i$ by $d$.
The vertex set $V$ is a subset of the $d \times d$ grid $\{(i,j) \mid 0 \leq i,j < d\}$.
We have $(i,j) \in V$ iff $(i+j) \mathop{\text{rem}} d \in \{0,1,\ldots,r-1\}$.
For example, for $r=4$ and $d=6$, the vertex set is arranged as follows (row $0$ is at the top, column $0$ is at the left):

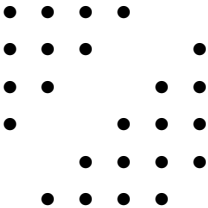
\begin{tikzpicture}[scale=0.5]
\node at (0,0) {$\bullet$};
\node at (1,0) {$\bullet$};
\node at (2,0) {$\bullet$};
\node at (3,0) {$\bullet$};
\node at (5,-1) {$\bullet$};
\node at (0,-1) {$\bullet$};
\node at (1,-1) {$\bullet$};
\node at (2,-1) {$\bullet$};
\node at (4,-2) {$\bullet$};
\node at (5,-2) {$\bullet$};
\node at (0,-2) {$\bullet$};
\node at (1,-2) {$\bullet$};
\node at (3,-3) {$\bullet$};
\node at (4,-3) {$\bullet$};
\node at (5,-3) {$\bullet$};
\node at (0,-3) {$\bullet$};
\node at (2,-4) {$\bullet$};
\node at (3,-4) {$\bullet$};
\node at (4,-4) {$\bullet$};
\node at (5,-4) {$\bullet$};
\node at (1,-5) {$\bullet$};
\node at (2,-5) {$\bullet$};
\node at (3,-5) {$\bullet$};
\node at (4,-5) {$\bullet$};
\end{tikzpicture}

Note that every row and every column has exactly $r$ boxes.
The rows correspond to the hyperedges of the first layer, where the columns correspond to the hyperedges of the second layer.
Note that no matter how the hyperedges of the third layer are placed, no two vertices can share all three hyperedges,
because no two vertices even share their two hyperedges in layer one and two.
Therefore an arbitrary placement of the third layer shows $t_{\delta,\delta}^\lambda > 0$.

For $d < r$ an analogous construction can be made, but several vertices share a location, see the example $r=6$ and $d=4$ below.

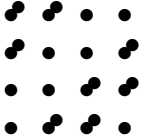
\begin{tikzpicture}[scale=0.5]
\node at (0,0) {$\bullet$};
\node at (0.2,0.2) {$\bullet$};
\node at (1,0) {$\bullet$};
\node at (1.2,0.2) {$\bullet$};
\node at (2,0) {$\bullet$};
\node at (3,0) {$\bullet$};
\node at (0,-1) {$\bullet$};
\node at (0.2,-0.8) {$\bullet$};
\node at (1,-1) {$\bullet$};
\node at (2,-1) {$\bullet$};
\node at (3,-1) {$\bullet$};
\node at (3.2,-0.8) {$\bullet$};
\node at (0,-2) {$\bullet$};
\node at (1,-2) {$\bullet$};
\node at (2,-2) {$\bullet$};
\node at (2.2,-1.8) {$\bullet$};
\node at (3,-2) {$\bullet$};
\node at (3.2,-1.8) {$\bullet$};
\node at (0,-3) {$\bullet$};
\node at (1,-3) {$\bullet$};
\node at (1.2,-2.8) {$\bullet$};
\node at (2,-3) {$\bullet$};
\node at (2.2,-2.8) {$\bullet$};
\node at (3,-3) {$\bullet$};
\end{tikzpicture}

Note that in this construction, if $r/d > 2$, then three or more vertices lie at the same position.
As in the case $d \geq r$, the rows correspond to the hyperedges of the first layer, where the columns correspond to the hyperedges of the second layer.
But now the third layer cannot be placed arbitrarily, but care has to be taken.
The hyperedges can be placed in any order, but not at arbitrary positions.
When a hyperedge is placed, it first uses those places where several vertices are grouped together (and of course only uses one from each such place). If there are places with more than two vertices, the hyperedge first takes vertices from those places with the most vertices.
This greedy method ensures that no hyperedge contains a pair of vertices from the same place, because by the length restriction on $\lambda$ a hyperedge cannot use more than $d^2$ vertices.
\end{proof}

\begin{acknowledge}
This work was supported by NSF grant CCF-1017760.
MW acknowledges support by the Simons Foundation, FQXI, and AFOSR (grant no.\ FA9550-16-1-0082).
A part of this work was done at the Simons Institute for the Theory of Computing, Berkeley.
\end{acknowledge}

\bibliography{main}
\end{document}